\documentclass[11pt,letterpaper]{article}

\usepackage[shortcuts]{extdash}

\title{The Effect of Sparsity on $k$-Dominating Set\\ and Related First-Order Graph Properties}

\author{Nick Fischer \\ Weizmann Institute of Science \thanks{This work is part of the project CONJEXITY that has received funding from the European Research Council (ERC) under the European Union's Horizon Europe research and innovation programme (grant agreement No.~101078482).}
\and Marvin Künnemann \\ Karlsruhe Institute of Technology \thanks{Research supported by the Deutsche Forschungsgemeinschaft (DFG, German 
Research
Foundation) – 462679611.} \and Mirza Redzic \\ Karlsruhe Institute of Technology \footnotemark[2]}
\date{}

\usepackage[mathscr]{eucal}
\usepackage[fleqn]{amsmath}
\usepackage{amssymb,bbm,amsthm}
\usepackage{thm-restate,thmtools}
\usepackage{subcaption}
\usepackage[utf8]{inputenc}
\usepackage{graphicx}
\usepackage{hyperref}

\newcommand{\bigO}{\mathcal{O}}

\usepackage{comment,color}
\usepackage{mathtools}

\usepackage{amsthm}
\usepackage{float}
\usepackage{cancel}
\usepackage[linesnumbered,boxed]{algorithm2e}
\usepackage[margin=1in, papersize={8.5in,11in}]{geometry}

\usepackage{enumitem}

\newcommand{\F}{\mathbb{F}}

\SetCommentSty{mycommfont}

\declaretheorem[numberwithin=section]{theorem}
\declaretheorem[unnumbered, name=Theorem]{theorem*}
\declaretheorem[numberlike=theorem]{lemma}

\declaretheorem[numberlike=theorem]{proposition}

\declaretheorem[numberlike=theorem]{corollary}

\declaretheorem[numberlike=theorem]{claim}

\declaretheorem[numberlike=theorem]{observation}

\declaretheorem[numberlike=theorem, name=Definition]{definition}
\declaretheorem[unnumbered, name=Definition]{definition*}
\declaretheorem[numberlike=theorem, name=Conjecture]{conjecture}
\declaretheorem[unnumbered, name=Conjecture]{conjecture*}

\declaretheorem[unnumbered, name=Hypothesis]{hypothesis*}

\newcommand{\tOh}{\tilde{\mathcal{O}}}

\newcommand{\FOP}[1]{$\mathrm{FOP}_k$}

\makeatletter
\renewcommand\paragraph{%
  \@startsection{paragraph}
    {4}
    {\z@}
    {3.25ex \@plus1ex \@minus.2ex}
    {-1em}
    {\normalfont\normalsize\bfseries\addperiod}}
\newcommand{\addperiod}[1]{#1\@addpunct{.}}

\DeclareMathOperator\MM{MM}
\def\Order{\bigO}

\begin{document}

\maketitle

\begin{abstract}
\noindent
We revisit the classic $k$-Dominating Set problem. Besides its importance as perhaps the most natural $W[2]$-complete problem, it is among the first problems for which a tight $n^{k-o(1)}$ conditional lower bound (for all sufficiently large $k$), based on the Strong Exponential Time Hypothesis (SETH), was shown (Pătraşcu and Williams, SODA 2007). Notably, however, the underlying reduction creates dense graphs, raising the question: how much does the sparsity of the graph affect its fine-grained complexity?

As our first result, we settle the fine-grained complexity of $k$-Dominating Set in terms of both the number of nodes $n$ and number of edges $m$, up to resolving the matrix multiplication exponent $\omega$. Specifically, on the hardness side, we show an~\makebox{$mn^{k-2-o(1)}$} lower bound based on SETH, for any dependence of $m$ on $n$. On the algorithmic side, this is complemented by an~\makebox{$mn^{k-2+o(1)}$}-time algorithm for all sufficiently large $k$. For the smallest non-trivial case of~\makebox{$k=2$}, i.e., 2-Dominating Set, we give a randomized algorithm that employs a Bloom-filter inspired hashing to improve the state of the art of $n^{\omega+o(1)}$ to $m^{\omega/2+o(1)}=O(m^{1.187})$. If $\omega=2$, this yields a conditionally tight bound for all $k\ge 2$.

To study whether $k$-Dominating Set is special in its sensitivity to sparsity, we study the effect of sparsity on very related problems:
\begin{itemize}
	\item The $k$-Dominating Set problem belongs to a type of first-order definable graph properties that we call \emph{monochromatic basic problems}. These problems are the canonical monochromatic variants of the basic problems that were proven complete for the class FOP of first-order definable properties (Gao, Impagliazzo, Kolokolova, and Williams, TALG 2019). We show that among the  monochromatic basic problems, the $k$-Dominating Set property is the \emph{only} property whose fine-grained complexity decreases in sparse graphs. Only for the special case of reflexive properties is there an additional basic problem that can be solved faster than $n^{k\pm o(1)}$ on sparse graphs.
	\item For the natural variant of distance-$r$ $k$-dominating set, we obtain a hardness of $n^{k-o(1)}$ under SETH for every $r\ge 2$ already on sparse graphs, which is tight for sufficiently large~$k$.
\end{itemize}

\end{abstract}

\setcounter{page}{0}
\thispagestyle{empty}
\newpage

\section{Introduction}
\label{sec:intro}

Consider an algorithmic graph problem whose best known algorithm runs in time $\tOh(n^c)$, where~$n$ denotes the number of vertices and $c\ge 2$ is some (usually small) constant. While improving this running time might resist intensive effort and even suffer from a conditional lower bound, it might still be possible to obtain substantial, \emph{polynomial-factor} improvements by taking into account the sparsity of the given graph -- after all, far from all interesting graphs are dense. A typical target to shoot for is a time bound of $\tOh(m^{c/2})$, where $m$ denotes the number of edges in the graph, or \emph{(edge) sparsity}. Such a running time is never worse than the known bound for dense graphs, i.e., $m=\Theta(n^2)$, but polynomially improves the running time for all sparser graphs, i.e., when~\makebox{$m=\bigO(n^{2-\epsilon})$}.

On the algorithmic side, obtaining such an algorithm may be technically challenging (consider, e.g., the global min-cut problem~\cite{KargerS96, Karger00}) or even turn out to be conditionally impossible: E.g., turning the $n^{\omega+o(1)}$-time state-of-the-art algorithm for All-Edge Triangle Detection to an $m^{\omega/2+o(1)}$-time algorithm would refute the 3SUM and APSP hypotheses~\cite{Patrascu10,VassilevskaWX20}.

While most works on the fine-grained complexity of graph problems analyze the time complexity either purely in $n$ (we call this the \emph{dense} case) or purely in $m$ (we call this the \emph{sparse} case), some recent works even specifically address the full trade-off between $m$ and $n$, e.g.~\cite{AgarwalR18, LincolnVWW18}: Agarwal and Ramachandran~\cite{AgarwalR18} give several sparsity-preserving reductions from the shortest cycle problem, giving evidence of optimality of natural $\tOh(mn)$-time algorithms. Subsequently, Lincoln et al.~\cite{LincolnVWW18} even manage to prove that the weighted $k$-clique hypothesis implies optimality of the $\tOh(mn)$-time bound for shortest cycle, specifically for $m=\Theta(n^\gamma)$ for infinitely many $1\le \gamma \le 2$. This yields conditional optimality for several APSP-related problems, such as directed and undirected APSP, radius, replacement paths, and more. Further related work addresses, e.g., the influence of sparsity for (unweighted) $k$-cycle detection~\cite{DahlgaardKS17,LincolnV20}.

In this paper we aim to advance this line of research by settling the effect of sparsity on interesting graph properties, most notably the $k$-Dominating Set problem.

\subsection{The Effect of Sparsity on \texorpdfstring{\boldmath$k$}{k}-Dominating Set}
\label{sec:effect-kdomSet}

We revisit the central graph problem $k$-Dominating Set for $k\ge 2$: Given an undirected graph $G=(V,E)$, determine whether there is a $k$-sized set $S$ of vertices such that each vertex $v\in V$ is \emph{dominated} by $S$ (i.e., $v\in S$ or there exists $u\in S$ with $\{u,v\}\in E$). It is among the classic NP-hard problems, counts as perhaps \emph{the} most natural $W[2]$-complete problem~\cite{DowneyF95} (see also~\cite{DowneyF13, CyganFKLMPPS15}) and suffers from strong fine-grained inapproximability results, see, e.g.~\cite{ChalermsookCKLM20,SLM19}. It gives rise to the notion of domination number in graph theory and has inspired a plethora of related problems (e.g., edge domination, total domination, partial domination, connected domination, capacitated domination and many more); see, e.g.,~\cite{HaynesHS98} for a dedicated monograph.

In the dense setting, the fine-grained complexity of $k$-Dominating Set is well understood (up to resolution of fast matrix multiplication):
Eisenbrand and Grandoni~\cite{EisenbrandG04} show how to solve $k$-Dominating Set in time $n^{k+o(1)}$ for all $k\ge 8$. For $2\le k \le 7$, the (small) polynomial overhead to the $\bigO(n^k)$ running time depends on the complexity of fast rectangular matrix multiplication. In particular, if $\omega=2$, $k$-Dominating Set can be solved in time $n^{k+o(1)}$ for all $k\ge 2$.

On the lower bound side, Pătraşcu and Williams~\cite{PatrascuW10} show that an $\bigO(n^{k-\epsilon})$ algorithm for any~\makebox{$k\ge 3$} would refute the Strong Exponential Time Hypothesis (SETH).\footnote{In fact, the lower bound can be based on the $k$-OV hypothesis, see Section~\ref{sec:prelims} for a definition.} Notably, this reduction creates dense graphs, and in particular does not give any lower bound for the case of $k=2$. We thus ask:

\begin{center}
\emph{Question 1: How does sparsity affect the time complexity of $k$-Dominating Set?
}    
\end{center}

Perhaps surprisingly, it has been observed in~\cite[Footnote 5]{BringmannFK19} that $k$-Dominating Set indeed admits faster than $n^{k\pm o(1)}$ algorithms in sparse graphs. The idea is simple: Any dominating set of size $k$ must contain at least one member that dominates at least $n/k$ vertices. Combining this idea with the algorithm of Eisenbrand and Grandoni~\cite{EisenbrandG04}, we obtain the following baseline, which depends on the optimal exponent $\omega(a,b,c)$ of multiplying a $n^a\times n^b$ matrix with a $n^b\times n^c$ matrix (and $\omega=\omega(1,1,1)$).

\begin{restatable}[$k$-Dominating Set Baseline]{proposition}{propbaseline} \label{prop:k-DomBaseline}
Let $k\ge 2$ and $1\le \gamma \le 2$. The $k$-Dominating Set on graphs with~$n$ nodes and $m=\Theta(n^\gamma)$ edges can be solved in time \smash{$n^{\omega+o(1)} + n^{\omega(\lceil \frac{k-1}{2} \rceil, 1, \lfloor \frac{k-1}{2} \rfloor + \gamma - 1) + o(1)}$}. If~\makebox{$\omega=2$}, this running time becomes $n^{2+o(1)} + mn^{k-2+o(1)}$.
\end{restatable}

For $k\ge 3$ (and assuming that $\omega=2$), this improves significantly over the running time for the dense case by a factor of $\Theta(n^2/m)$. However, under Pătraşcu and Williams' lower bound~\cite{PatrascuW10}, we could hope for even better speed-ups -- possibly even for an algorithm running in time $m^{k/2+o(1)}$. As our first contribution, we show that the $\Theta(n^2/m)$-improvement by Proposition~\ref{prop:k-DomBaseline} is best-possible, assuming the $k$-Orthogonal Vectors Hypothesis, which is well-known to be implied by the Strong Exponential Time Hypothesis (see Section~\ref{sec:prelims}).


\begin{restatable}[$k$-Dominating Set Lower Bound]{theorem}{thmdomsetlowerbound} \label{theorem:DomSet-lower-bound}
For all $k\ge 3$ and $\epsilon > 0$, there is no algorithm for $k$-Dominating Set in time $\bigO(mn^{k-2-\epsilon})$, unless the $k$-OV Hypothesis fails.
\end{restatable}

Curiously, Theorem~\ref{theorem:DomSet-lower-bound} leaves open an important special case, as it provides no non-trivial lower bound for $k=2$. At the same time, Proposition~\ref{prop:k-DomBaseline} also does not improve over the best known upper bound of $n^{\omega+o(1)}$ due to Eisenbrand and Grandoni~\cite{EisenbrandG04}. This brings us to an unclear situation: Can we improve the $n^{\omega+o(1)}$-time algorithm for 2-Dominating Set on sparse graphs, or can the lower bound from Theorem~\ref{theorem:DomSet-lower-bound} be strengthened?

Our main algorithmic result is that for 2-Dominating Set we indeed can improve upon Proposition~\ref{prop:k-DomBaseline}:

\begin{restatable}[$2$-Dominating Set Algorithm]{theorem}{thmtwodomsetalgo} \label{theorem:DomSet-algorithm}
    There is a randomized algorithm solving 2-Dominating Set in time $m^{\omega/2+o(1)}$.
\end{restatable}
Note that if $\omega = 2$, this yields an almost-optimal $m^{1+o(1)}$-time algorithm for 2-Dominating Set. More generally, if $\omega=2$, our results conditionally establish that $mn^{k-2+o(1)}$ is the \emph{optimal} running time for $k$-Dominating Set for all $k\ge 2$, up to subpolynomial factors.

We remark that if $\omega > 2$, our algorithm achieves an even better running time for very small graph densities (\smash{$m \leq n^{\frac{2+\omega}{4}} \le n^{1.094}$}). Specifically, our running time is never worse than $\tOh(m^2/n)$ and thus near-linear in for very sparse graphs with $m = \tOh(n)$. See Section~\ref{sec:k-dom-set} for more details.

\subsection{Beyond \texorpdfstring{\boldmath$k$}{k}-Dominating Set: Monochromatic First-Order Graph Properties}

The non-trivial influence of sparsity on the complexity of $k$-Dominating Set raises the question how general this phenomenon is:

\begin{center}
	\emph{Question 2: For which related graph problems does sparsity influence the time complexity?}
\end{center}

To approach this question systematically, we observe that $k$-Dominating Set is a first-order definable property of the following form: Given an undirected graph $G=(V,E)$ with $|V|=n$ and $|E|=m$, decide if
\[ \exists v_1\in V \dots \exists v_k \in V: \forall w \in V: E(v_1,w) \vee \cdots \vee E(v_k,w).\]
(In this formulation, we assume that the edge predicate is symmetric and reflexive, i.e., $E(v,v)$ for all $v\in V$.)

There are many interesting problems that may be formulated as such a first-order definable graph property, such as existence of a given $k$-vertex pattern, existence of a $k$-sized set of vertices sharing no common neighbor, the property of having a graph diameter 2, and many more. By allowing an even more general formulation\footnote{Specifically, allowing an arbitrary number of relations of arbitrary constant arity, not just a single edge relation.}, one arrives at the class FOP defined by Gao, Impagliazzo, Kolokolova and Williams~\cite{GaoIKW19}: For any first-order definable property $\phi$, FOP contains the corresponding problem of deciding $\phi$ over a given relational structure (in many interesting cases, simply a graph).

General algorithmic results for this class have been obtained by Williams~\cite{Williams14} for graph properties in the dense case (where we consider the universe size $n$ as main parameter) and by~\cite{GaoIKW19} for general properties in the sparse case (where we consider the total size of the relational structure $m$ as main parameter). 
Specifically, Williams showed that all $(k+1)$-quantifier first-order graph properties can be solved in time $\bigO(n^{k+o(1)})$ for $k\ge 8$ (which would even hold for all $k\ge 2$ if $\omega=2$) and additionally gave a SETH-based lower bound of $n^{k-o(1)}$ for some properties. Gao et al.~\cite{GaoIKW19} show that all $(k+1)$-quantifier properties can be solved in time $\bigO(m^k)$, and for each $k\ge 2$, determine a list of $(k+1)$-quantifier problems, called \emph{basic problems} of order $k$, to be \emph{complete} for this class in the following sense: An $\bigO(m^{k-\epsilon})$-time algorithm for any of these complete problems would give a polynomial improvement over the $\bigO(m^{k'})$-time algorithm for all problems with $k'+1$ quantifiers where $k'\ge k$. We shall call a problem that is complete in this sense an \emph{\FOP{k}-complete} problem. For any $k\ge 2$, the basic problems of order $k$ have the following form:
\begin{align*}
&\exists v_1 \in V_1 \dots \exists v_k \in V_k \forall w\in W: \ell_1  \vee \cdots \vee \ell_k,\\
&\text{where each literal } \ell_i \text{ is of the form } E(v_i,w) \text{ or } \overline{E(v_i,w)}
\end{align*}
Put differently, to obtain a basic problem, one must choose, for each $i\in [k]$, precisely one of $E(v_i,w)$ and its negation $\overline{E(v_i, w)}$. Note that this establishes the basic problems as fine-grained equivalent, \emph{hardest} problems in FOP. Among these basic problems, we find the $k$-OV problem (see also~\cite{GaoIKW19} for a detailed discussion), and a problem that is usually not formulated in graph language: $k$-Set Cover, see below.

\paragraph{Monochromatic vs.~Bichromatic: \texorpdfstring{\boldmath$k$}{k}-Dominating Set vs.~\texorpdfstring{\boldmath$k$}{k}-Set Cover} In the  $k$-Set Cover problem, the input consists of a set family $\mathcal{S}$ over universe $U$, and the question is whether there are $k$ sets from $\mathcal{S}$ that cover $U$, i.e.,
\[ \exists S_1\in {\cal S} \dots \exists S_k \in {\cal S}: \forall u  \in U: u \in S_1 \vee \cdots \vee u \in S_k.\]
By introducing a set $S_v \coloneqq N[v]$ for all $v\in V$, this problem generalizes the $k$-Dominating Set problem. In fact, $k$-Set Cover can be equivalently viewed as a \emph{bichromatic} version of $k$-Dominating Set (also known as Red-Blue Dominating Set): Define a 2-partite graph $G=({\cal S}\cup U, E)$ where for any $S_i \in \mathcal{S}, u\in U$, we have $\{S_i,u\}\in E$ if and only if $u\in S_i$. Then the task is to determine a set of $k$ vertices $S_1,\dots, S_k$ chosen from $\mathcal{S}$ such that $S_1,\dots,S_k$ dominate all vertices in $U$.

Pătraşcu and Williams observe that their $n^{k-o(1)}$ conditional lower bounds for $k$-Dominating Set extends to $k$-Set Cover. In fact, it is not difficult to see (and implicit in~\cite{PatrascuW10}) that the reduction for $k$-Dominating Set can be slightly simplified to establish hardness of $k$-Set Cover already for 
all $k\ge 2$ (rather than $k\ge 3$). Moreover, the hardness reduction produces sparse instances for which $\sum_{i=1}^n |S_i| = \bigO(n)$. Thus, our results in Section~\ref{sec:effect-kdomSet} \emph{separate} $k$-Set Cover (the bichromatic variant) from $k$-Dominating Set (the monochromatic variant), as the effect of sparsity differs for both problems.
This leads to the natural question whether monochromatic versions are always easier to solve on sparse graphs than on general graphs.

\paragraph{Monochromatic Basic Problems.} To address this question in some generality, we perform a comprehensive study on the \emph{monochromatic basic problems}, i.e., the canonical \emph{monochromatic} versions of the basic problems of FOP. Our monochromatic basic problems have the form:
\begin{align*}
&\exists \text{ pairwise distinct } v_1,\dots, v_k \in V \forall w\in V: \ell_1  \vee \cdots \vee \ell_k,\\
&\text{where each literal } \ell_i \text{ is of the form } E(v_i,w) \text{ or } \overline{E(v_i,w)}
\end{align*}
Note that here we introduce as additional requirement that the existentially quantified variables $v_1,\dots, v_k$ are pairwise distinct. For monochromatic properties, this requirement is indeed the usually intended meaning -- otherwise, most of the basic problems become trivial.\footnote{Specifically, whenever the basic problem contains both a positive disjunct $E(v_i,v)$ and a negative disjunct $\overline{E(v_j,v)}$, the property is trivially satisfied by choosing $v_i = v_j$.} Our monochromatic basic graph problems (of order $k$) contain several natural examples:
\begin{itemize}[itemsep=\smallskipamount,parsep=0pt,partopsep=0pt]
	\item\emph{$k$-Dominating Set:} Is there a subset of $k$ vertices dominating all vertices?
	\item\emph{Neighborhood Containment:} Are there distinct $v_1,v_2\in V$ such that $N(v_1) \subseteq N(v_2)$?
	\item\emph{Neighborhood $(k-1)$-Covering:} Is there a vertex whose neighborhood can be covered by the neighbors of $k-1$ other vertices?
	\item\emph{$k$-Empty-Neighborhood-Intersection:} Are there $k$ vertices that have no common neighbor?
	\item\emph{$(k-1)$-Common Neighborhood:} Is there a vertex whose neighbors are common neighbors of $k-1$ other vertices?
\end{itemize}

Although the (multichromatic) basic problems are all fine-grained equivalent (by being \FOP{k}-complete), we show that among the monochromatic basic problems, $k$-Dominating Set is surprisingly different:
\begin{theorem}[Basic Problems Lower Bound] \label{thm:monochromatic-except-kdomSet}
	Let $k\ge 2$ and $\epsilon > 0$. If any of the monochromatic basic problems \emph{except $k$-Dominating Set} can be solved in time $\bigO(n^{k-\epsilon})$ on graphs with $n^{1+o(1)}$ edges, then the $k$-OV Hypothesis is false.
\end{theorem}

This establishes that $k$-Dominating Set is the \emph{only} monochromatic basic problem that becomes easier on sparse graphs, answering our driving Question 2.

Interestingly, this result also shows the fine-grained equivalence of almost all basic problems to their monochromatic versions -- proving such multichromatic-to-monochromatic reductions for simpler first-order properties, such as detection of certain patterns of size $k$, would contradict established hardness assumptions in fine-grained complexity theory. As a case in point, one can show that the 4-chromatic version of 4-cycle detection conditionally requires $n^{\omega-o(1)}$ time under the triangle detection hypothesis, while the monochromatic version of 4-cycle detection is solvable in time $\bigO(n^2)$~\cite{YusterZ94}.

\paragraph{Special Cases: Reflexivity vs.~Irreflexivity}
We highlight an additional reason why Theorem~\ref{thm:monochromatic-except-kdomSet} appears surprising, as it does \emph{not} hold for the special case of \emph{reflexive} properties. Specifically, our definition of monochromatic basic properties allows for the existence of self-loops, i.e., $E(v,v)$ with $v\in V$, to be specified individually for each $v\in V$, as part of the input. It may be reasonable to either disallow self-loops (i.e., require $\overline{E(v,v)}$ for all $v\in V$; we call this the \emph{irreflexive case}) or enforce self-loops (i.e., require $E(v,v)$ for all $v\in V$; we call this the \emph{reflexive} case). The reflexive special case generally expresses problems over closed neighborhoods and the irreflexive special case expresses problems over open neighborhoods.

Our proof of Theorem~\ref{thm:monochromatic-except-kdomSet} establishes the same hardness for sparse graphs in the \emph{irreflexive} case. For the \emph{reflexive} case, however, it turns out that there exists an additional basic problem for which $n^{k\pm o(1)}$ time can be broken for sparse graphs: Closed Neighborhood $(k-1)$-Covering. In this problem the task is to detect distinct $v_1,\dots, v_k$ vertices such that the closed neighborhood of $v_1$ is covered by the closed neighborhoods of $v_2,\dots, v_k$, i.e., $N[v_1]\subseteq N[v_2]\cup \cdots \cup N[v_k]$. We design an algorithm running in time $m^{\frac{2\omega}{\omega +1}}n^{k-2} \le \bigO(m^{1.41}n^{k-2})$. This beats running time $n^{k\pm o(1)}$ whenever $m \leq n^{1+\frac{1}{\omega}-o(1)}$ (the exponent is roughly $1.42$ under the current value of $\omega$).
To our surprise, we further establish that in the reflexive case, $k$-Dominating Set and Closed Neighborhood $(k-1)$-Covering are the \emph{only} monochromatic basic problems that are influenced by sparsity.

\paragraph{Distance-\texorpdfstring{\boldmath$r$}{r} Domination}
A popular generalization of $k$-Dominating Set is the Distance-$r$ $k$-Dominating Set problem, see, e.g.,~\cite{EickmeyerGKKPRS17,KreutzerRSW17, BergBKMZ20}. In this variant, a vertex $v$ is dominated by $S\subseteq V$, if there exists some $s\in S$ with distance at most $r$ from $v$. In Section~\ref{sec:distance-r-domination}, we show that Distance-$r$ $k$-Dominating Set is affected by sparsity if and only if $r=1$: For $r=1$, we obtain the usual $k$-Dominating Set problem with the complexity $mn^{k-2}$ established in Section~\ref{sec:effect-kdomSet} if $\omega=2$. For any $r\ge 2$, we prove a hardness of $n^{k-o(1)}$ under SETH already in sparse graphs, which is tight for all $k\ge 2$ if $\omega=2$.

\paragraph{Related Work}
Investigating the fine-grained complexity of classes of first-order definable problems has recently gained traction, see, e.g.,~\cite{Williams14, GaoIKW19, BringmannFK19,BringmannCFK21,BringmannCFK22,AnGIJKP22}.

Establishing hardness results for monochromatic settings generally appears to be technically challenging: For additive problems, specifically 3-Linear Degeneracy Testing, \cite{DudekGS20} exploit involved constructions from additive combinatorics ($k$-sum-free sets) to establish the equivalence of monochromatic and multichromatic variants. Another example is the geometric setting of Closest Pair in the Euclidean Metric, for which a fine-grained equivalence between the bichromatic and monochromatic case could be shown~\cite{SM20}.

%




\subsection{Technical Overview}

We give an outline of our most interesting technical ideas. Specifically, we sketch the algorithmic improvements for $k$-Dominating Set for small values of $k$, as well as our general reduction from multichromatic basic properties except $k$-Dominating Set to their monochromatic versions. All further contributions are detailed in their respective technical sections.

\paragraph{Algorithmic Contributions for \boldmath$k$-Dominating Set}

To exploit sparsity for $k$-Dominating Set, the first crucial observation is that in any $k$-dominating set $\{v_1,\dots, v_k\}$ there must exist a node $v_i$ of degree at least $n/k-1$. Let $H$ denote the set of such nodes; we clearly have that $|H| \le \frac{2m}{(n/k - 1)} = \bigO(m/n)$. Thus, we may restrict our search for a dominating set to the search space $H \times V^{k-1}$ of size $\bigO(mn^{k-2})$.

However, naively testing each set in $H \times V^{k-1}$ still requires an overhead of $\bigO(n)$ per candidate solution. In this way we cannot beat running time $n^{k\pm o(1)}$, so we have to be more careful. It is only natural to try to adapt the approach of Eisenbrand and Grandoni~\cite{EisenbrandG04}. Let us consider the case of~\makebox{$k=2$}: For any sets $S, T\subseteq V$, we denote by $A_{S,T}$ the adjacency matrix of $G$ restricted to $S\times T$ (i.e., $A_{S,T}$ is the 0-1 matrix whose rows are indexed by $S$, whose columns are indexed by $T$, and whose entries are defined by $A_{S,T}[s,t] = 1$ iff $\{s,t\}\in E$ or $s=t$). Furthermore, let $\overline{A}$ denote the complement of $A$. Then it holds that
\[ (\overline{A_{H,V}} \cdot \overline{A_{V,V}})[v_1,v_2] = 0 \text{ iff } N[v_1]\cup N[v_2] = V,\]
 i.e., $v_1,v_2$ form a 2-dominating set. This reduces 2-dominating set to the multiplication of a rectangular $\bigO(m/n)\times n$ matrix $\overline{A_{H,V}}$ with a square $n\times n$ matrix $\overline{A_{V,V}}$.

Since already the input size for this matrix product is of size $\Theta(n^2)$, this approach again cannot directly achieve a $m^{1+o(1)}$-time algorithm. To avoid this, one might hope to use techniques for sparse matrix multiplication (see, e.g.,\cite{AlonYZ97,YusterZ05}), since the adjacency matrix of a sparse graph has~$\bigO(m)$ nonzeroes. However, the \emph{non}-adjacency matrix of a sparse graph as required here necessarily has~$\Theta(n^2)$ nonzeroes. Fortunately, we can still formulate the problem as a sparse matrix multiplication: Specifically, we have
\begin{equation}\label{eq:specialcase}
(\overline{A_{H,V}} \cdot A_{V,V})[v_1,v_2] = \overline{\deg(v_1)} \text{ iff } N[v_1]\cup N[v_2] = V,
\end{equation}
where $\overline{\deg(v)} = n - \deg(v)$ denotes the number of nonzeroes of the row corresponding to $v$ in $\overline{A_{H,V}}$. Note that the number of nonzeroes of $\overline{A_{H,V}}$ (as an $\bigO(m/n)\times n$ matrix) and $A_{V,V}$ (as adjacency matrix of an $m$-edge graph) is $\bigO(m)$. Using sparse matrix multiplication/triangle counting, we can solve this problem in time $\bigO(m^{\frac{\omega}{2\omega + 1}})$~\cite{AlonYZ97}. Still, this does not yield linear-time complexity even if $\omega=2$ -- in particular, current sparse matrix multiplication techniques fail to beat $n^{2\pm o(1)}$ running time even for multiplying a $\sqrt{n} \times n$ matrix with a sparse $n\times n$ matrix containing $\bigO(n\sqrt{n})$ nonzeroes~\cite{AmossenP09}.

Our crucial contribution is that, perhaps surprisingly, we can compute the special case outlined by~\eqref{eq:specialcase} faster than computing the full matrix product $B\coloneqq\overline{A_{H,V}} \cdot A_{V,V}$. To this end, we partition~$V$ into logarithmically many groups $V_i$ ($0 \leq i \leq \lfloor \log n \rfloor$) consisting of all nodes with degree in~\makebox{$[2^i, 2^{i+1})$}. For each $v_2\in V_i$, we observe that it can only form a 2-dominating set with some $v_1\in H$ if $\overline{\deg(v_1)} \le \deg(v_2) < 2^{i+1}$.  This simple observation allows us to employ a Bloom-filter-like approach: Let $H_i$ denote set of nodes $v_1\in H$ with $\overline{\deg(v_1)} \le 2^{i+1}$. We construct hash functions \makebox{$h_1,\dots, h_L : V\to [O(2^i)]$} with $L=\bigO(\log n)$ to reduce the inner dimension of the matrix multiplication $\overline{A_{H_i,V}} \cdot A_{V,V_i}$ to size $L\cdot \bigO(2^i)$, such that with high probability, any entry in the result matrix corresponding to $v_1\in H, v_2\in V_i$ is equal to the number of nonzeroes in $v_1$'s row if and only if $v_1,v_2$ are a 2-dominating set. That this is possible is due to the special structure of~\eqref{eq:specialcase} (and would fail for more general decision problems for sparse matrix products). In total, we perform $\bigO(\log n)$ multiplications -- namely, for each $0 \leq i \leq \lfloor \log n\rfloor$, we multiply a $\bigO(m/n)\times \tOh(2^i)$ by a $\tOh(2^i) \times \min\{\frac{m}{2^i}, n\}$ matrix. This can be shown to take time~\smash{$m^{\omega/2+o(1)}$}. In particular, this algorithm runs in almost-linear time~\smash{$m^{1+o(1)}$} if~\makebox{$\omega=2$}. We give all details in Section~\ref{sec:k-dom-set}.

\paragraph{Hardness for Monochromatic Properties}
Recall that for the class of $k$-Dominating Set-like problems, the monochromatic basic problems, we prove that any problem other than $k$-Dominating Set conditionally requires time $n^{k-o(1)}$. For the sketch of this proof, let us consider as a simple exemplary property, the Neighborhood Containment problem: \makebox{$\exists \text{distinct } v_1,v_2\in V: N(v_1) \subseteq N(v_2)$}. Already for this simple problem, we face many technical challenges.

As for all monochromatic basic problems, there is a known hardness reduction (from the $k$-OV problem) to its multichromatic variant $\exists v_1\in V_1 \exists  v_2\in V_2 \forall w\in W: E(v_1,w) \Rightarrow E(v_2,w)$. A natural attempt would be to add auxiliary vertices $V_A$ that enforce any solution $(v_1,v_2)$ with $N(v_1)\subseteq N(v_2)$ to be chosen from $V_1\times V_2$. However, great care has to be taken for these auxiliary vertices: any node in $V_A$ could be chosen itself as $v_1$ or $v_2$. Furthermore, if \emph{any} node in $V_1$ is connected to some auxiliary node $z\in V_A$, then \emph{all} nodes in $V_2$ need to be connected to $z$. To keep the whole instance sparse, this requires adding only a very small number of auxiliary nodes with carefully chosen connections. The situation gets more intricate for more complicated properties such as $\exists \text{ pairwise distinct } x_1,\dots, x_4\in V: N(v_1)\cap N(v_2) \subseteq N(v_3) \cup N(v_4)$, where no two auxiliary nodes $z_1,z_2$ may have disjoint neighborhoods, since otherwise we could set $v_1=z_1,v_2=z_2$ and get a trivial solution with any $v_3,v_4$. For proving hardness, this shows that no gadget can act fully locally, but must take into account the full graph.

To nevertheless prove hardness for all monochromatic properties except $k$-Dominating Set, we proceed via an intermediate step of bichromatic properties. For any $\exists^k\forall$-quantified first-order property $\phi$, we distinguish between its \emph{multichromatic}, \emph{bichromatic} and \emph{monochromatic} versions defined as follows:
\begin{itemize}[itemsep=\smallskipamount]
	\item \emph{Multichromatic:} $\exists v_1\in V_1 \dots \exists v_k\in V_k \,\forall w\in W \,\phi(v_1, \dots, v_k, w)$.
	\item \emph{Bichromatic:} $\exists \text{ pairwise distinct } x_1,\dots, x_k \in  X \,\forall y\in Y \,\phi(x_1, \dots ,x_k, y)$.
	\item \emph{Monochromatic:} $\exists \text{ pairwise distinct } v_1,\dots, v_k \in  V \,\forall v\in V \,\phi(v_1,\dots ,v_k, v)$.
\end{itemize}

\paragraph{Step 1: From Multichromatic to Bichromatic}
The first step is to reduce a multichromatic problem to its bichromatic version: We show how to construct, for any subset $S\subseteq X$, a \emph{solution-excluding gadget} $Y_S$ (with corresponding edges) such that no choice of pairwise distinct vertices~\makebox{$x_1,\dots, x_k\in S$} can satisfy $\forall y\in Y_S: \phi(x_1,\dots, x_k, y)$. Let us call a variable $x_i$ a \emph{positive} variable if $E(x_i,y)$ occurs in $\phi$ and a \emph{negative} variable if $\overline{E(x_i,y)}$ occurs in $\phi$. Furthermore,  let $\phi(x)$ be a unique identifier for $x\in S$ consisting of $\bigO(\log n)$ bits. The main idea is that for every pair of a positive variable $x_s$ and a negative variable $x_j$ we can find a bit position \smash{$i_j^{(s)} \in [O(\log n)]$} where their identifiers $\phi(x_s)$ and $\phi(x_j)$ differ. For every guess of such bit positions and the corresponding bit values $\phi(x_j)[i_j^{(s)}]$, we introduce a corresponding node in $Y_S$ and connect it in such a way to $S$ such that for every choice of pairwise distinct $x_1,\dots, x_k$, the node for the correct guess is adjacent to all negative variables and non-adjacent to all positive variables. The number of nodes in $Y_S$ is at most $\bigO(\log^k n)=\tOh(1)$.

Equipped with this tool, we can create a bichromatic instance as follows: We set $X=V_1\cup \cdots \cup V_k$. To obtain $Y$, we start with $W$, and include, for each subset $S_i = X\setminus V_i$, $i\in [k]$, a \emph{solution-excluding gadget} $Y_{S_i}$ with possibly additional edges to $V_i$, enforcing that the only way to satisfy all nodes in~$Y_{S_i}$ is to pick $x_i\in V_i$.


\paragraph{Step 2: From Bichromatic to Monochromatic}
It remains to reduce the bichromatic to the monochromatic version, at least for the case that $Y$ is of size $\tOh(1)$; note that the reduction from $k$-OV to the bichromatic setting as sketched above indeed maintains $|Y|=\tOh(1)$, so this is sufficient for our purposes. Interestingly, there is a rather simple randomized reduction based on the probabilistic method, and a technically more interesting derandomization.
(This is a common phenomenon for fine-grained reductions involving coding-theory-like gadgets.)
We proceed as follows: Given a bichromatic instance $G=(X\cup Y, E) $ with $Y= \{y_1, \dots, y_M\}$, we aim to construct an equivalent monochromatic instance $G'=(V',E')$. By a \emph{(partial) solution} $P$, we understand a choice for (a subset of) the variables $v_1,\dots, v_k$ (formally, it would be a mapping $\{v_1,\dots, v_k\}\to V\cup \{*\}$, where $*$ denotes an unspecified variable -- by abuse of notation, we view it as a subset $P\subseteq V$, which hides that we need to distinguish between positive and negative variables). We say that $P$ satisfies $v\in V$ if $\phi(v_1,\dots,v_k,v)$ is satisfied by at least one of the assigned variables in $P$.

The idea is to construct a graph $H$ with $M^{f(k)}=\tOh(1)$ nodes such that (1) each node $v\in V(H)$ receives a label $\lambda(v)\in \{1,\dots, M\}$, and (2) for any solution $P$ of size $k$ and every node label $i\in \{1,\dots, M\}$, there exists a node $v\in V(H) \setminus P$ with $\lambda(v)=i$ that is \emph{not satisfied} by $P\cap V(H)$.

Equipped with such an object, we can construct $G'$ by setting $V'=X\cup V(H)$ and adding an edge (1) between $x\in X$ and $v\in V(H)$ iff $x$ and $y_{\lambda(v)}$ were adjacent in $G$, and (2) between $v,v'\in V(H)$ iff $v$ and $v'$ were adjacent in $H$. Note that there are no edges within $X$.

We claim that any satisfying solution $x_1,\dots, x_k$ in $G$ yields a corresponding solution $x_1,\dots,x_k$ in $G'$: Since $\phi$ has at least one negative literal, we have that $\phi(x_1,\dots, x_k,x)$ is satisfied for every $x\in X$.\footnote{This is the argument that crucially fails for the $k$-dominating set property.} We claim that also $\phi(x_1,\dots,x_k,v)$ is satisfied for all $v\in V(H)$, since $\phi(x_1,\dots, x_k,v_{\lambda(v)})$ is satisfied in $G$.

Conversely, we claim that if $G$ has no satisfying solution, then also $G'$ has no satisfying solution. Consider any solution $P$ in $G'$. Since $G$ has no satisfying solution, the partial solution $X\cap P$ must leave at least some vertex $v_i\in Y$ in $G$ unsatisfied. By construction of $H$, there exists some node $\tilde{v}\in V(H) \setminus P$ with label $i$ such that $P\cap V(H)$ does not yet satisfy $\tilde{v}$. However, $\tilde{v}$ also cannot be satisfied by $P\cap X$ by definition of $i$. Thus, $\tilde{v}$ is left unsatisfied by the full solution $P$.

Constructing the graph $H$ via a randomized algorithm is not too difficult using the probabilistic method. However, we are even able to give an explicit, deterministic construction:
The rough idea is to identify the nodes of $H$ with polynomials of degree $d$ over the finite field of size $p$, for suitably chosen $p,d$. Any such polynomial $f$ can be equivalently viewed by its evaluations $(f(1),f(2),\dots, f(d+1))$. We will use appropriate parameters $L < d  < R$ (specifically, $L=d-k$ and $R=d+dk+1$) and identify $(f(1),\dots, f(L))$ with the label $\lambda(f)$ and define an edge between $f$ and $g$ iff there exists some $L < x \le R$ such that $f(x)=g(x)$. Now for any degree-$d$ polynomials $a_1,\dots, a_k$ and $b_1,\dots, b_k$, and any label $(\ell_1,\dots, \ell_L)$, we can prove existence of some degree-$d$ polynomial $f$ with $f(i) = \ell_i$ for all $i\in [L]$, as well as  $\{f,a_j\}\in E$ for all $j\in [k]$ and $\{f,b_j\}\notin E$ for all $j\in [k]$. To do this, one must find $(d+1)-L$ additional evaluations for $f$ such that for each $j\in [k]$ there exists some $L < x \le R$ such that $f(x)=a_j(x)$ while at the same time for each $j\in [k]$, we have $f(x) \ne b_j(x)$ for \emph{all} $L < x \le R$. We give all details in Section~\ref{sec:monochromatic-props}.

\section{Preliminaries}
\label{sec:prelims}
Let $n$ be a positive integer. 
We denote by $[n]$ the set $\{1,\dots, n\}$. 
If $S$ is an $n$-element set and $0\leq k\leq n$ is an integer, then $\binom{S}{k}$ denotes the set of all $k$-element subsets of $S$.
We denote by $\mathcal{P}(S)$ the power set of $S$.

We use $\Tilde{\bigO}$ notation, which hides the poly-logarithmic factors. 
In other words, $f(n)\in \Tilde{\bigO}(g(n))$ if and only if there exists a $k\in \bigO(1)$ such that $f(n)\in \bigO(g(n)\cdot \log^k(n))$.

Let $\omega<2.3729$ \cite{AlmanW21} denote the optimal exponent of multiplying two $n\times n$ matrices and $\omega(a,b,c)$ denote the optimal exponent for multiplying an $n^a\times n^b$ matrix by an $n^b\times n^c$ matrix.

Let $G$ be a graph and $X\subseteq V(G)$.
Then we denote by $G[X]$ the subgraph of $G$ induced by $X$.
If $G,H$ are isomorphic, we write $G\cong H$.
For any vertex $v\in V(G)$, the \emph{neighbourhood} of $v$ is the set of vertices adjacent to $v$, denoted $N(v)$. 
The \emph{closed neighbourhood} of $v$, denoted $N[v]$ is defined as $N[v]:=N(v)\cup\{v\}$.
The \emph{degree} of $v$ denotes the size of its neighbourhood ($\deg(v) = |N(v)|$).
For any two vertices $u,v\in V(G)$, we denote by $d_G(u,v)$ the length of the shortest path between $u$ and $v$ in $G$.

\paragraph{Hypotheses}
Consider the $k$-Orthogonal Vectors problem ($k$-OV) that is stated as follows.
Given $k$ sets $A_1,\dots, A_{k}$ of $d$-dimensional binary vectors, decide whether there exist vectors $a_1\in A_1,\dots, a_k\in A_k$ such that for all $t\in [d]$, it holds that $\prod_{i=1}^{k} a_i[t] = 0$.
A simple brute force approach solves the $k$-OV in time $\bigO(d\cdot \prod_{i\in [k]}|A_i|)$.

On the other hand, it is well known that an algorithm solving $k$-OV with $|A_1| = \dots = |A_k| = n$ and $d=\log^2 n$ in time $\bigO(n^{k-\varepsilon})$ would refute SETH\footnote{$d=\log^2 n$ can be replaced by any $d=\omega(\log n)$.}, which follows by combining a split-and-list reduction~\cite{PatrascuW10} with the sparsification lemma~\cite{ImpagliazzoPZ01}, see~\cite{VassilevskaW18} for details. This conjecture is known as \emph{(low-dimensional) $k$-OV Hypothesis}.
For the purpose of this paper, we consider a more general formulation where we allow the sets to be of different sizes and hence state the $k$-OVH as follows.
\begin{conjecture}[$k$-OVH]
    For no $\varepsilon>0$ and for no $0\leq \gamma_1,\dots,\gamma_k\leq 1$ is there an algorithm solving $k$-OV with $|A_1| = n^{\gamma_1}, \dots, |A_k| = n^{\gamma_k}$, $d=\log^2 n$ in time $\bigO(n^{(\sum_{i=1}^{k}\gamma_i)-\varepsilon})$.
\end{conjecture}
We refer to the setting of $k$-OVH with $\gamma_1 = \dots = \gamma_k = 1$ as \emph{balanced} $k$-OVH. 

These two hypotheses are known to be equivalent (see \cite[Lemma II.1]{BringmannFOCS15} for a proof for $k=2$).
Below we give a proof for general $k$.
\begin{lemma}
    Balanced $k$-OVH and $k$-OVH are equivalent.
\end{lemma}
\begin{proof}
    $k$-OVH implies balanced $k$-OVH trivially.
    Conversely, we show that refuting $k$-OVH refutes the balanced $k$-OVH.
    To this end, assume that for some $0\leq \gamma_1,\dots, \gamma_k\leq 1$ there exists an algorithm $\mathcal{A}$ solving the $k$-OV with $|A_1| = n^{\gamma_1},\dots, |A_k| = n^{\gamma_k}$ in time $\bigO(n^{(\sum_{i=1}^{k}\gamma_i)-\varepsilon})$.

    Then given an instance of balanced $k$-OV with $|A_1| = \dots = |A_k| = n$, we can partition each of the sets $A_i$ into $s := \lceil n^{1-\gamma_i}\rceil$ subsets $A_i^{1},\dots, A_i^{s}$ and run the algorithm $\mathcal{A}$ on each combination of subsets and return true if for at least one instance the algorithm $\mathcal{A}$ returns true and false otherwise. 
    Clearly, this yields a correct algorithm for the balanced $k$-OV and it runs in time $\bigO(n^{1-\gamma_1}\cdots n^{1-\gamma_k}\cdot n^{(\sum_{i=1}^{k}\gamma_i)-\varepsilon}) = \bigO(n^{k-\varepsilon})$.
\end{proof}
\section{Algorithms and Hardness of \texorpdfstring{\boldmath$k$}{k}-Dominating Set}
\label{sec:k-dom-set}
In this section we provide our algorithms and conditional hardness results for the $k$-Dominating Set problem. We start by recalling the baseline algorithm (Proposition~\ref{prop:k-DomBaseline}) in Section~\ref{sec:k-dom-set:sec:baseline}. Then, in Section~\ref{sec:k-dom-set:sec:2-dom-set} we develop our improved algorithm for the $2$-Dominating Set problem in sparse graphs, and in Section~\ref{sec:k-dom-set:sec:lower-bound} we strengthen the conditional lower bounds for $k$-Dominating Set to match our algorithms.

\subsection{The Baseline Algorithm} \label{sec:k-dom-set:sec:baseline}
For the baseline algorithm (and also for our later improvements) we rely on the following simple observation:
\begin{observation} \label{obs:heavy-vertex}
    Given a graph $G$, let $S = \{v_1,\dots, v_k\}$ be a dominating set of $G$.
    Then, there exists $1\leq i \leq k$, such that $\deg(v_i)\geq \frac{n}{k}-1$.
\end{observation}
We call any such vertex $v\in V(G)$ with $\deg(v)\geq \frac{n}{k}-1$ a \emph{high-degree} vertex.
Observe that in any graph $G$ with $n$ vertices and $m$ edges, there are at most $\bigO(\frac{m}{n})$ high-degree vertices (assuming~$k$ is a fixed constant).
For the rest of this section, let $G$ denote a graph with $n$ vertices and $m=n^\gamma$ edges for some $1\leq \gamma\leq 2$ and let $H\subseteq V$ denote the set of high-degree vertices.

\propbaseline*
\begin{proof}
    Let $\mathcal{S}\subseteq \mathcal{P}(V)$ and denote by $B_\mathcal{S}$ the binary matrix whose rows are indexed by $\mathcal{S}$ and whose columns are indexed by $V$ and the entry $B_\mathcal{S}[S,v] = 1$ if and only if $S$ dominates~$v$ (i.e., $v\in S$ or there exists $s\in S$, such that $\{s,v\}\in E$).
    Similarly, for $\mathcal{T}\subseteq \mathcal{P}(V)$ let $C_\mathcal{T}$ denote the matrix whose rows are indexed by $V$ and whose columns are indexed by $\mathcal{T}$ and the entry $C_\mathcal{T}[v,T] = 1$ if and only if $T$ dominates $v$. Clearly, we have $(\overline{B_\mathcal{S}}\cdot \overline{C_\mathcal{T}})[S,T] = 0$ if and only if $S\cup T$ dominate $G$.

    Let $H\subseteq V$ be the set of all high degree vertices and let $\mathcal S$ and $\mathcal{T}$ be defined as follows:
    \begin{equation*}
        \mathcal S =  \binom{V}{\lceil \frac{k-1}{2}\rceil}, \; \mathcal T = \big \{T \in \binom{V}{\lfloor \frac{k-1}{2}\rfloor + 1} \mid T\cap H \neq \emptyset \big \}.
    \end{equation*}
    Recall that any $k$-dominating set $S^*$ contains at least one high-degree vertex (Observation~\ref{obs:heavy-vertex}). We can therefore partition $S^*$ into two subsets $S^* = S \cup T$, where $S$ has size~\smash{$\lceil{\frac{k-1}{2}}\rceil$}, $T$ has size~\smash{$\lfloor\frac{k-1}{2}\rfloor + 1$} and $T$ contains the high-degree vertex. It follows that the graph contains a $k$-dominating set if and only if the matrix product $\overline{B_\mathcal{S}}\cdot \overline{C_\mathcal T}$ contains a zero entry.

    It remains to analyze the running time of computing this matrix product. Observe that
    \begin{equation*}
        |\mathcal S| = \bigO(n^{\lceil \frac{k-1}{2}\rceil}),\; |\mathcal T| = \bigO\left(\frac{m}{n}\cdot n^{\lfloor \frac{k-1}{2}\rfloor}\right) = \bigO(n^{\gamma - 1 + \lfloor \frac{k-1}{2}\rfloor}).
    \end{equation*}
    Thus, $B_{\mathcal S}$ is an $\bigO(n^{\lceil{\frac{k-1}{2}}\rceil}) \times n$ matrix and $C_{\mathcal T}$ is an $n \times \bigO(n^{\gamma-1+\lfloor\frac{k-1}{2}\rfloor})$ matrices, and computing their product takes time $\bigO(n^{\omega+o(1)} + n^{\omega(\lceil \frac{k-1}{2} \rceil, 1, \lfloor \frac{k-1}{2} \rfloor + \gamma - 1) + o(1)})$ as claimed.
    %
    %
\end{proof}

\subsection{A Faster Algorithm for \texorpdfstring{\boldmath$2$}{2}-Dominating Set} \label{sec:k-dom-set:sec:2-dom-set}
In this section we design our improved algorithm for the $2$-Dominating Set for sparse graphs. Specifically, our goal is to prove the following theorem:\footnote{By employing the state-of-the-art fast rectangular matrix multiplication techniques (e.g.~\cite{GallU18}), a more fine-grained analysis of the algorithm reveals that we can achieve even slightly better running time in very sparse graphs.}

\thmtwodomsetalgo*

Our strategy is to phrase our 2-Dominating Set algorithm as an algorithmic \emph{reduction} to the following intermediate problem: 

\begin{definition}[Max-Entry Matrix Product]\label{def:max-entry}
Consider $0$-$1$-matrices $B$ of size $\Order(\frac mn) \times n$ and~$C$ of size $n \times n$, where~$C$ contains at most $m$ nonzeros. The \emph{Max-Entry Matrix Product} problem is to decide whether there exist $i, j$ such that
\begin{equation*}
    (B \cdot C)[i, j] = \sum_k B[i, k] \cdot C[k, j] = \sum_k B[i, k].
\end{equation*}
\end{definition}

We remark that, when viewing $B$ and $C$ as the bi-adjacency lists of a tripartite graph, the Max-Entry Matrix Product problem asks whether there is a pair $i, j$ of outer vertices such that every edge from $i$ can be extended to a 2-path to $j$. Another equivalent formulation is in terms of the \emph{Subset Query} problem: Given families $\mathcal S, \mathcal T$ of subsets of some universe $[U]$, the goal is to test whether there exist sets $S \in \mathcal S$ and $T \in \mathcal T$ such that $S \subseteq T$. The Max-Entry Matrix Product problem is exactly the special case of Subset Query where $U = n$, $|\mathcal S| = \Order(\frac{m}{n})$ and $\sum_{T \in \mathcal T} |T| = \Order(m)$. While the (unrestricted) Subset Query problem has been studied in previous works~\cite{GaoIKW19,CharikarIP02}, we decided to stick to the matrix version from Definition~\ref{def:max-entry} which is more in line with our view on the $k$-dominating set problem as explained in the overview.

The Max-Entry Matrix Product problem can naively be solved in time $\Order(m \cdot \frac{m}{n})$ (enumerate an index $i$ and a nonzero entry $C[k, j])$ and in time $\MM(\frac{m}{n}, n, n)$ by fast matrix multiplication. As we will prove later, by a more elaborate algorithm the Max-Entry Matrix Product problem can be solved in time $m^{\omega/2+o(1)}$.

Before, the first step towards our algorithm is to reduce $2$-Dominating Set on sparse graphs to Max-Entry Matrix Product. The reduction is somewhat similar to the baseline algorithm in the previous section. 

\begin{lemma}[2-Dominating Set to Max-Entry Matrix Product] \label{lemma:domset-sparse-mm-reduction}
If the Max-Entry Matrix Product problem can be solved in time $T(n, m)$, then the 2-Dominating Set problem can be solved in time $\bigO(T(n, m))$.
\end{lemma}
\begin{proof}
Let $G = (V, E)$ be a given 2-Dominating Set instance. Let $A$ be the adjacency matrix of~$G$, where we understand that $A[i,i] = 1$ for all $i$. Moreover, let $H \subseteq V$ denote the subset of high-degree vertices (with degree at least $\frac n2 - 1$), and let $\overline{A_{H,V}}$ denote the complement of the adjacency matrix restricted to~\makebox{$H \times V$}. We claim that $\{u, v\}$ is a 2-dominating set of $G$ if and only if $(\overline{A_{H,V}} \cdot A)[u,v] = \deg_{\overline{G}}(u)$.

Indeed, note that $(\overline{A_{H,V}} \cdot A)[u,v] = \sum_{w\in V} \overline{A_{H,V}}[u,w]\cdot  A[w,v]$ counts the number of vertices~$w$ that are not dominated by $u$, but are dominated by $v$. Furthermore, $u\in H$ and $v\in V$ form a 2-dominating set if and only if every node $w$ that is not dominated by $u$ is dominated by $v$. Since there are precisely $\deg_{\overline{G}}(u)$ nodes that are non-dominated by $u$, we conclude that $u$ and $v$ form a 2-dominating set if and only if $(\overline{A_{H,V}} \cdot A)[u,v]=\deg_{\overline{G}}(u) = \sum_{w} \overline{A_{H,V}}[u, w]$.

Picking $B = \overline{A_{H, V}}$ (which has size $\Order(\frac mn) \times n$) and $C = A$ (which has size $n \times n$ and is $m$-sparse), we have therefore successfully reduced to an instance of Max-Entry Matrix Product. Constructing these matrices runs in linear time and is therefore negligible.
\end{proof}

We proceed to the core of our algorithm.

\begin{lemma}[Max-Entry Matrix Product to Rectangular Matrix Multiplication] \label{lem:2-dom-set-rectangular}
There is a randomized algorithm solving Max-Entry Matrix Product in time
\begin{equation*}
    \widetilde\bigO\left(\max_{1 \leq d \leq n} \MM\!\left(\frac{m}{n},\, d,\, \min\left\{n, \frac md\right\}\right)\right).
\end{equation*}
\end{lemma}
\begin{proof}
For an index $j$, the \emph{degree} $\sum_k C[k, j]$ denotes the number of nonzero entries~$C[k, j]$.
As a first step, we will split $C$ column-wise into $L = \lfloor\log n\rfloor$ many submatrices $C_0, \dots, C_L$ such that all degrees in $C_\ell$ are in the range $[2^\ell, 2^{\ell+1})$. Each such submatrix still has size at most $n \times n$, and contains at most $m$ nonzero entries. Even better: Since each column in $C_\ell$ contains at least $2^\ell$ nonzero entries, there can be at most $\frac{m}{2^\ell}$ columns in $C_\ell$, and thus $C_\ell$ has size at most $n \times \min\{n, \frac{m}{2^\ell}\}$. For the rest of the proof, fix any $0 \leq \ell \leq L$; we solve the Max-Entry Matrix Product problem on~$(B, C_\ell)$ (for simplicity we will often omit the subscript from $C_\ell$).

Our goal is to reduce Max-Entry Matrix Product to a regular matrix product. Unfortunately, while the outer dimensions are comparably small ($\frac mn$ and $\min\{n, \frac{m}{2^\ell}\}$), the inner dimension can be as large as $n$. The key idea is to apply a Bloom-filter-like construction to also compress the inner dimension. To this end, sample a uniformly random hash function $h : [n] \to [2^{\ell+2}]$ and consider the matrices $B'$ (of size $\Order(\frac mn) \times 2^{\ell+2}$) and $C'$ (of size $2^{\ell+2} \times \min\{n, \frac{m}{2^\ell}\}$) defined by
\begin{align*}
    B'[i, k'] &= \bigvee_{\substack{k \in [n]\\h(k) = k'}} B[i, k], \\
    C'[k', j] &= \bigvee_{\substack{k \in [n]\\h(k) = k'}} C[k, j].
\end{align*}
We claim that these matrices preserve solutions in the following sense:

\begin{claim} \label{lem:2-dom-set-rectangular:clm:soundness}
If $(B \cdot C)[i, j] = \sum_k B[i, k]$, then $(B' \cdot C')[i, j] = \sum_{k'} B'[i, k']$.
\end{claim}
\begin{proof}
Assume that $(B \cdot C)[i, j] = \sum_k B[i, k] \cdot C[k, j] = \sum_k B[i, k]$, and take any $k'$ with $B'[i, k'] = 1$. By construction, there is some $k \in [n]$ with $h(k) = k'$ and $B[i, k] = 1$. Our initial assumption implies that also $C[k, j] = 1$. But then our construction assigns $C'[k', j] = 1$. Since $k'$ was arbitrary, it follows that indeed $(B' \cdot C')[i, j] = \sum_{k'} B'[i, k'] \cdot C[k', j] = \sum_{k'} B[i, k']$.
\end{proof}

\begin{claim} \label{lem:2-dom-set-rectangular:clm:completeness}
If $(B \cdot C)[i, j] < \sum_k B[i, k]$, then $(B' \cdot C')[i, j] < \sum_{k'} B'[i, k']$ with probability at least~$\frac12$.
\end{claim}
\begin{proof}
From the assumption that $(B \cdot C)[i, j] < \sum_k B[i, k]$, it follows that there is some $k^* \in [n]$ with $B[i, k^*] = 1$ and $C[k^*, j] = 0$. We claim that with probability at least $\frac12$, there is no index~\makebox{$k \neq k^*$} such that $h(k) = h(k^*)$ and $C[k, j] = 1$. Indeed, for any fixed $k \neq k^*$ the event $h(k) = h(k^*)$ happens with probability at most $\frac{1}{2^{\ell+2}}$. Moreover, in the $j$-th row of $C$ there are at most $2^{\ell+1}$ nonzero entries $C[k, j]$, thus by a union bound the error probability is bounded by \smash{$\frac{2^{\ell+1}}{2^{\ell+2}} = \frac12$}.
\end{proof}

This suggests the following algorithm: Repeat, for $R = 100 \log n$ iterations, the construction of~$B'$ and~$C'$ (with fresh randomness). We solve the Max-Entry Matrix Product problem on $(B', C')$ using fast matrix multiplication. If there is a pair $(i, j)$ such that across all repetitions we have $(B' \cdot C')[i, j] = \sum_{k'} B'[i, k']$, then we report ``yes''. Otherwise, we report ``no''. By Claim~\ref{lem:2-dom-set-rectangular:clm:soundness}, we will always correctly report ``yes'' instances. In a ``no'' instance, the probability that we mistakenly report ``yes'' due to some false positive $(i, j)$ is bounded by $2^{-R} \leq n^{-100}$ by Claim~\ref{lem:2-dom-set-rectangular:clm:completeness}. Taking a union bound over the at most $n^2$ pairs, the total error probability is at most $n^{-98}$.

Let us finally consider the running time. For a fixed $\ell$, we compute $\Order(\log n)$ matrix products of size $\Order(\frac mn) \times 2^{\ell+1} \times \min\{n, \frac m{2^\ell}\}$. This takes time $\widetilde\bigO(\MM(\frac mn, 2^\ell, \min\{n, \frac{m}{2^\ell}\}))$. It takes linear time to construct the matrices $B'$ and $C'$, respectively, and this overhead is negligible in the running time. Summing over the $\log n$ levels $0 \leq \ell \leq L = \lfloor{\log n}\rfloor$, the total time is
\begin{equation*}
    \sum_{\ell=0}^{\lfloor\log n\rfloor} \widetilde\bigO\left(\MM\left(\frac mn, 2^\ell, \min\left\{n, \frac{m}{2^\ell}\right\}\right)\right) = \max_{1 \leq d \leq n} \widetilde\bigO\left(\MM\left(\frac mn, d, \min\left\{n, \frac{m}{d}\right\}\right)\right),
\end{equation*}
which is as claimed.
\end{proof}

This completes the description of the Max-Entry Matrix Product algorithm; however, it remains to carefully analyze the complicated running time expression:

\begin{corollary}[Max-Entry Matrix Product] \label{cor:max-entry}
There is a randomized algorithm solving Max-Entry Matrix Product in time $m^{\omega/2+o(1)}$.
\end{corollary}
\begin{proof}
We use the algorithm from Lemma~\ref{lem:2-dom-set-rectangular}. It runs in time $\widetilde\bigO(\max_{1 \leq d \leq n} T(n, m, d))$ where
\begin{equation*}
    T(n, m, d) = \MM\!\left(\frac{m}{n},\, d,\, \min\left\{n, \frac md\right\}\right),
\end{equation*}
and it remains to bound $T(n, m, d) \leq m^{\frac{\omega}{2}+o(1)}$ for any choice of $1 \leq d \leq n$. Throughout this proof, we will only use that $\MM(N, N, N) \leq N^{\omega+o(1)}$ (by fast square matrix multiplication) and that trivially $\MM(N_1 \cdot N_1', N_2, N_3) = \bigO(N_1') \cdot \MM(N_1, N_2, N_3)$ (similarly for $N_2$ and $N_3$). We distinguish the following three cases, each of which can be proved by a calculation:
\begin{itemize}
    \item \emph{Case 1: $n \leq \frac{m}{d}$:}
    \begin{align*}
        T(n, m, d) &\leq \MM(\tfrac mn, \tfrac mn, n) && (d \leq \tfrac{m}{n}) \\
        &\leq \bigO(\tfrac{n^2}{m}) \cdot \MM(\tfrac mn, \tfrac mn, \tfrac mn) && (n^2 \geq m) \\
        &\leq \Order(\tfrac{n^2}{m}) \cdot (\tfrac{m}{n})^{\omega+o(1)} \\
        &\leq \tfrac{m^{\omega-1+o(1)}}{n^{\omega-2}} \\
        &\leq m^{\omega/2+o(1)}. && (n^2 \geq m)
    \end{align*}
    \item \emph{Case 2: $n \geq \frac{m}{d}$ and $d \leq \frac{m}{d}$:}
    \begin{align*}
        T(n, m, d) &\leq \MM(d, d, \tfrac{m}{d}) && (n \geq \tfrac md) \\
        &\leq \Order(\tfrac{m}{d^2}) \cdot \MM(d, d, d) && (d \leq \tfrac md) \\
        &\leq \Order(\tfrac{m}{d^2}) \cdot d^{\omega+o(1)} \\
        &\leq m \cdot d^{\omega-2+o(1)} \\
        &\leq m \cdot \sqrt m^{\omega-2+o(1)} && (d \leq \tfrac md) \\
        &\leq m^{\omega/2+o(1)}.
    \end{align*}
    \item \emph{Case 3: $n \geq \frac{m}{d}$ and $d \geq \frac{m}{d}$:}
    \begin{align*}
        T(n, m, d) &\leq \MM(\tfrac{m}{d}, d, \tfrac{m}{d}) &&(d \leq n, n \geq \tfrac md) \\
        &\leq \Order(\tfrac{d^2}{m}) \cdot \MM(\tfrac{m}{d}, \tfrac{m}{d}, \tfrac{m}{d}) &&(d \geq \tfrac md) \\
        &\leq \Order(\tfrac{d^2}{m}) \cdot (\tfrac{m}{d})^{\omega+o(1)} \\
        &\leq \tfrac{m^{\omega-1+o(1)}}{d^{\omega-2}} \\
        &\leq \tfrac{m^{\omega-1+o(1)}}{\sqrt{m}^{\omega-2}} &&(d \geq \tfrac md)  \\
        &\leq m^{\omega/2+o(1)}.
    \end{align*}
\end{itemize}
In all three cases, we have successfully bounded the running time by $m^{\omega/2+o(1)}$. This completes the proof.
\end{proof}



We have assembled all parts of the proof of Theorem~\ref{theorem:DomSet-algorithm}: We combine the reduction from Lemma~\ref{lemma:domset-sparse-mm-reduction} with the efficient Max-Entry Matrix Product algorithm from Corollary~\ref{cor:max-entry}. It follows that the 2-Dominating Set problem is in time $m^{\omega/2+o(1)}$.

Recall that the baseline algorithm (Proposition~\ref{prop:k-DomBaseline}) in combination with our 2-Dominating Set algorithm (Theorem~\ref{theorem:DomSet-algorithm}) prove that for $\omega = 2$, the $k$-Dominating Set problem can be solved in time~\smash{$m n^{k-2+o(1)}$} for all $k \geq 2$. In the next section, we complement this upper bound by a matching conditional lower bound.

\subsection{A Matching Conditional Lower Bound}\label{sec:k-dom-set:sec:lower-bound}
In this section our goal is to prove the following theorem:
\thmdomsetlowerbound*

The idea is to adapt the reduction by Pătraşcu and Williams~\cite{PatrascuW10} to create a sparse instance.
In particular, we reduce from the unbalanced variant of $k$-OV problem obtained by setting $|A_1| = \dots = |A_{k-1}| = n$ and $|A_k| = n^{\gamma-1}$ for the suitable $1\leq \gamma \leq 2$, and then we use the vertices corresponding to vectors in $A_k$ to dominate the remaining sets. 
\begin{figure}
	\centering
	\includegraphics{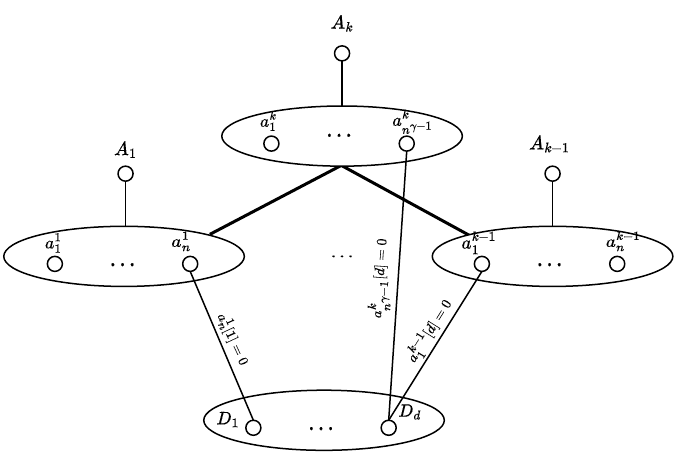}
	\caption{Schematic representation of the graph $G$ as constructed in proof of Theorem \ref{theorem:DomSet-lower-bound}. Bold edges represent bicliques between the corresponding sets.}
	\label{fig:k-ds-lower}
\end{figure}
\begin{proof}
    To prove this lower bound, we reduce from $k$-Orthogonal Vectors problem.

    Let $A_1,\dots, A_{k-1}$ denote sets of $d$-dimensional (assume $d=\bigO(\log^2 n)$) binary vectors of size $n$ and $A_k$ be a set of $d$-dimensional binary vectors of size $n^{\gamma-1}$ for $1\leq \gamma\leq 2$.
    We construct a graph $G$ with $\Tilde{\bigO}(n)$ vertices and $\Tilde{\bigO}(n^{\gamma})$ edges that has dominating set of size $k$ if and only if we can find vectors $a_1\in A_1,\dots, a_k\in A_k$ such that $\prod_{i=1}^k a_i[j] = 0$ for every $j = 1,\dots, d$.

    Consider the following construction illustrated in Figure \ref{fig:k-ds-lower}.
    Let $V(G)$ consist of the $k$ vertices corresponding to sets $A_1,\dots, A_k$, the $\bigO(n)$ vertices corresponding to the vectors, labeled $a^i_j\in A_i$, and the vertices labeled $D_1,\dots, D_d$, corresponding to the $d$ dimensions. 
    For each $1\leq i \leq k$ add all edges between $A_i$ and $a^i_j$ for all $j$.      
    For all $1\leq i\leq k-1$ add an edge between every vertex corresponding to a vector contained in $A_i$ (namely all the vertices $a^i_j$) and all the vertices corresponding to vectors in $A_{k}$
    Add an edge between the vertices $a^i_j$ and $D_t$ if and only if $a^i_j[t] = 0$ for all $i,j,t$.

    The graph $G$ clearly has $\bigO(n + d) = \bigO(n)$ vertices and observe that there are $\bigO(n^{\gamma-1} + d) = \bigO(n^{\gamma-1})$ vertices with degree at most $\bigO(n)$ (in particular the vertices corresponding to vectors in $a^k_i$ for all $i\in [n^{\gamma-1}]$, vertices $D_1,\dots, D_d$ and the vertices $A_i$ for all $i\in [k]$). 
    All other vertices have degree at most $n^{\gamma-1}$, hence the total number of edges is at most $\bigO(n\cdot n^{\gamma-1} + dn) = \Tilde{\bigO}(n^{\gamma})$ and consequently, the graph $G$ can be constructed in time $\bigO(n^\gamma d) = \Tilde{\bigO}(n^{\gamma})$.
    
    It is now sufficient to argue that $G$ has a dominating set of size $k$ if and only if there exist vectors $a_1\in A_1, \dots, a_k\in A_k$ such that $a_1[i]\cdot \ldots \cdot a_k[i] = 0$ for every $i\in [d]$.
    To this end, suppose first that such vectors exist.
    Then we show that the vertices corresponding to $a_1,\dots, a_k$ form a dominating set.
    
    First observe that the vertex corresponding to $a_i$ is adjacent to the vertex $A_i$, and hence the vertices $A_1,\dots, A_k$ are dominated.
    Furthermore, since all vertices $a^i_j$ (for $1\leq i < k$) are adjacent to all the vertices corresponding to vectors in $A_k$, clearly all such vertices are dominated by $a_k$ and symmetrically all vertices $a^k_j$ are dominated by e.g. $a_1$. 

    It is only left to verify that the vertices $D_1,\dots, D_d$ are dominated.
    Consider arbitrary $D_i$.
    Now, since $a_1[i]\cdot \dots \cdot a_k[i] = 0$, there exists some $1\leq j\leq k$, such that $a_j[i] = 0$ and consequently, by construction of $G$, $a_j$ is adjacent to $D_i$.

    Conversely, suppose that $G$ assumes a dominating set $S = \{s_1,\dots, s_k\}$ of size $k$. 
    Consider first the vertices $A_1,\dots, A_k$.
    Clearly, their closed neighbourhoods are disjoint, hence each vertex $s_i$ from $S$ is contained in $N[A_i]$ (without loss of generality).
    Furthermore, since $\{D_1,\dots, D_d\}\cap S = \emptyset$, and the vertices $D_t$ are non-adjacent to the vertices $A_1,\dots, A_k$, there exists at least one $i\in [k]$, such that $s_i = a^i_j$ for some $j$.
    
    We construct the solution of $k$ orthogonal vectors as follows.
    For each $s_i$ such that $s_i = A_i$, let $a_i$ be an arbitrary vector from $A_i$.
    For the remaining vertices let $a_i$ be the vector corresponding to $s_i$.
    We claim that $a_1\in A_1,\dots, a_k\in A_k$ satisfy $a_1[i]\cdot \ldots \cdot a_k[i] = 0$ for all $i\in [k]$.
    To see this, consider an arbitrary $i$.
    Observe that $D_i$ is dominated by some $s_j$, and in particular $s_j\neq A_j$. 
    Hence $a_j[i] = 0$ and the claim follows.
\end{proof}

\section{Hardness of Monochromatic Basic Problems}
\label{sec:monochromatic-props}

In this section, we will prove hardness of all monochromatic basic problems except $k$-Dominating Set in sparse graphs. We proceed in three steps: (1) First, we recall hardness established in~\cite{GaoIKW19} for all multichromatic basic problems, which already holds in sparse graphs. (2) We then reduce each multichromatic basic problem to its bichromatic version (where one of the parts maintains a size of $\tOh(1)$). (3) Finally, we reduce each such bichromatic problem except bichromatic $k$-Dominating Set (i.e., $k$-Set Cover) to its monochromatic version.

For completeness, we begin with a proof that all Multichromatic Basic Problems are $k$-OV hard even in sparse graphs.
\begin{proposition}[Implicit in~\cite{GaoIKW19}]
    Let $\Phi$ be a basic property of order $k$.
    Given an instance $A_1,\dots, A_k\subseteq\{0,1\}^d$ of $k$-OV, we can construct in time $\Tilde\bigO(n)$ an equivalent instance $G = (V_1,\dots, V_k, W, E)$ of the Multichromatic Basic Problem for $\Phi$ with at most $\Tilde{\bigO}(n)$ edges and $|W|=\tOh(1)$.
\end{proposition}
\begin{proof}
    Without loss of generality, assume that $\Phi= E(v_1,w)\lor \cdots \lor E(v_\ell, w) \lor \overline{E(v_{\ell+1},w)}\lor \cdots \lor \overline{E(v_k,w)}$.
    We call a vertex $v_i$ \emph{positive} (resp.\ \emph{negative}) if $\Phi$ contains the literal $E(v_i,w)$ (resp. $\overline{E(v_i,w)}$).
    For each dimension $i$, add a vertex $w_i$ to $W$ and let $V_i = A_i$ for all $i\in [k]$.
    Now for every vertex $v\in V_1,\dots, V_{\ell}$, add an edge between $v$ and $w_i$ if and only if the vector corresponding to $v$ has $i$-th entry equal to zero.
    Similarly, for every vertex $v\in V_{\ell + 1},\dots, V_{k}$, add an edge between $v$ and $w_i$ if and only if the vector corresponding to $v$ has $i$-th entry equal to one.

    If $A_1,\dots, A_k$ is a yes-instance of $k$-OV, then we can find vectors $a_1\in A_1,\dots, a_k\in A_k$ such that for any dimension $j$, some vector $a_i$ satisfies $a_i[j] = 0$.
    We proceed to show that choosing $v_1\in V_1, \dots, v_k\in V_k$ such that $v_i$ is the vertex corresponding to $a_i$ satisfies $\Phi$.
    If $1\leq i\leq \ell$, then $\Phi$ contains the literal $E(v_i,w)$ and the corresponding vertex $v_i$ is by construction adjacent to $w_j$ satisfying $\Phi$.
    Otherwise, $\Phi$ contains the literal $\overline{E(v_i,w)}$ and the corresponding vertex from $v_i$ is non-adjacent to $w_j$, again satisfying $\Phi$.

    On the other hand, if $A_1,\dots, A_k$ is a no-instance of $k$-OV, then for every selection of vectors $a_1\in A_1,\dots, a_k\in A_k$ there exists a dimension $j$, such that $a_i[j] = 1$.
    In particular, if we choose any $v_1\in V_1,\dots, v_k\in V_k$, then considering the corresponding vectors yields a vertex $w_j\in W$ that is non-adjacent to all $v_1,\dots, v_\ell$ and adjacent to all $v_{\ell+1},\dots, v_k$, thus leaving $\Phi(v_1,\dots, v_k,w_j)$ unsatisfied. 

    Note that $|W|=d=\tOh(1)$. Since every edge in $G$ connects $v\in V_i$ with $w\in W$, the total number of edges is at most $\bigO((\sum_{i=1}^{k}|V_i|)\cdot|W|) = \bigO(nd) = \Tilde{\bigO}(n)$.
\end{proof}

\subsection{Multichromatic to Bichromatic}
In order to prove the hardness of Monochromatic Basic Problems with at least one negative literal in the sparse graphs, we introduce the intermediate class of \emph{Bichromatic Basic Problems} defined as follows.
Let $k$ be a fixed constant and for every $i\in [k]$, let $L_i:X\times Y \to \{0,1\}$ be either defined as $L_i(x,y) = E(x,y)$, or $L_i(x,y) = \overline{E(x,y)}$.
Now all the Bichromatic Basic Problems can be stated in the following way.
Given a bipartite graph $G = (X,Y,E)$ decide if there is a set of pairwise distinct vertices $x_1,\dots, x_k\in X$, such that for all $y\in Y$ it holds that $L_1(x_1,y)\lor\dots \lor L_k(x_k,y)$?

We proceed to prove that all the Bichromatic Basic Problems are $k$-OV hard already in sparse graphs, by reducing from the corresponding Multichromatic Basic Problems.
To this end, for any subset $S\subseteq X$, we first show how to construct an \emph{identifier gadget} $J_S$, such that if we add corresponding edges between $S$ and $J_S$ then for every $x_1,\dots,x_k \in S$ there exists $y\in J_S$ such that $\phi(x_1,\dots,x_k,y)$ is not satisfied.
In order to do this, we first need to introduce some tools.

Given a set $S$ with $n$ elements, let $\phi:S\to \{0,1\}^{\log(n)}$ be an injective function called the \emph{identifier} of $S$. 
Let $x$ be a binary vector with $d$ dimensions and define
$$f(x) = \Big \{\big(i_1^{(1)},\dots, i_k^{(1)},\dots,i_1^{(k)},\dots i_k^{(k)}\big), \big(b_1^{(1)}, \dots, b_k^{(1)},\dots,b_1^{(k)},\dots, b_k^{(k)}\big ) \mid \bigwedge_{s=1}^k \big(\bigvee_{t=1}^k x[i_t^{(s)}] = b_t^{(s)}\big)\Big \}$$
Let $S$ be a set of $n$ elements and let
$$T = \Big\{\big(x,I,B\big) \mid x\in S, (I,B) \in f(\phi(x)) \Big\}$$
For the rest of this section, let $S,T,f,\phi$ be as above. 
The following lemma shows that for any subset $S'$ consisting of at most $2k$ elements from $S$, by labeling at most $k$ distinct elements in $S'$ \emph{positive} and at most $k$ distinct elements in $S'$ \emph{negative} (such that no element is labeled both positive and negative), we can always find the the tuple of corresponding indices $I$ and bits $B$, such that $(x,I,B)$ is in $T$ for each positive element $x$ and $(y,I,B)$ is not in $T$ for any negative element $y$.
This will be a key observation allowing us to build the gadget that has the desired properties.
\begin{lemma}
    For every subset $\{x_1,\dots, x_k, y_1, \dots, y_k\}\subseteq S$, there exist indices $I$, and bits $B$ such that $(x_i,I, B)\in T$ and $(y_i,I, B)\not \in T$ for all $i$.
\end{lemma}
\begin{proof}
    Intuitively, for any fixed $y_s$, we want to find indices $i_1^{(s)},\dots, i_k^{(s)}$ and bits $b_1^{(s)},\dots, b_k^{(s)}$ that leave the clause $\bigvee_{t=1}^{k}\phi(y_s)[i_t^{(s)}] = b^{(s)}_t$ unsatisfied, and observe that this is sufficient to show that for arbitrary selection of the remaining indices and bits, $(y_s,I,B)$ is not in $T$. 
    Simultaneously, we want to make sure that for every $x_t$ the clause $\bigvee_{t=1}^{k}\phi(x_t)[i_t^{(s)}] = b^{(s)}_t$ is satisfied, by assuring that $\phi(x_t)[i_t^{(s)}] = b^{(s)}_t$ for every $t$.
    To this end, we exploit the injectiveness of $\phi$.
    In particular, since $\phi$ is injective, we can find an index $i_1^{(1)}$, such that $\phi(x_1)[i_1^{(1)}]\neq \phi(y_1)[i_1^{(1)}]$.
    Similarly, for all $2\leq j \leq k$, find indices $i_2^{(1)},\dots, i_k^{(1)}$, such that $\phi(x_j)[i_j^{(1)}]\neq \phi(y_1)[i_j^{(1)}]$.
    
    Let $b^{(1)}_t = \phi(x_t)[i_t^{(1)}]$ for all $1\leq t \leq k$ and observe that for all $t\in [k]$, $\phi(y_1)[i_t^{(1)}] \neq b^{(1)}_t$, and hence the formula $\bigvee_{t=1}^{k}\phi(y_1)[i_t^{(1)}] = b^{(1)}_t$ is not satisfied.
    On the other hand, clearly for any $j$,
    the formula $\bigvee_{t=1}^{k}\phi(x_j)[i_t^{(1)}] = b^{(1)}_t$ is satisfied since $\phi(x_j)[i_j^{(1)}] = b^{(1)}_j$.
    Similarly, for all $2\leq s \leq k$, one can find indices $i_1^{(s)},\dots, i_k^{(s)}$, such that $\phi(x_j)[i_j^{(s)}]\neq \phi(y_s)[i_j^{(s)}]$ and symmetrically setting $b^{(s)}_t = \phi(x_t)[i_t^{(s)}]$ leaves the formula $\bigvee_{t=1}^{k}\phi(y_s)[i_t^{(s)}] = b^{(s)}_t$ unsatisfied, while satisfying $\bigvee_{t=1}^{k}\phi(x_j)[i_t^{(s)}] = b^{(s)}_t$.
    
    Finally, we observe that consequently the formula $\bigwedge_{s=1}^k \big(\bigvee_{t=1}^k \phi(x_j)[i_t^{(s)}] = b_t^{(s)}\big)$ is satisfied for every $j$ and $\bigwedge_{s=1}^k \big(\bigvee_{t=1}^k \phi(y_j)[i_t^{(s)}] = b_t^{(s)}\big)$ is not satisfied for any $j$. 
    The desired result follows.
\end{proof}
We proceed to use the result from the previous lemma to construct our gadget. 
In particular, the vertices in our gadget will correspond to the indices and the bits as above and we will add an edge between a vertex $x$ from the set $S\subseteq X$ and the vertex $(I,B)$ in the gadget if and only if $(x,I,B)$ is in $T$. 
By the previous lemma, for any selection of at most $k$ positive and at most $k$ negative vertices, we can find a vertex in the gadget that is adjacent to all negative vertices and non-adjacent to all positive vertices, thus appending the gadget vertices to the set $Y$ implies that no complete solution is contained inside $S$. 
\begin{lemma}[Identifier gadget] \label{lemma:identifier-gadget}
    Let $G = (X,Y,E)$ be a bipartite graph and $k$ a fixed constant. 
    If $S\subseteq X$, then one can add $\Tilde{\bigO}(1)$ vertices to $Y$ and connect those vertices to $S$ so that for any $k$ vertices $x_1,\dots, x_k\in S$ and for any $k$ vertices $x'_1,\dots, x'_k\in X$ there exists a vertex $y\in Y$ such that $y\in N(x_1)\cap N(x_2)\cap \dots \cap N(x_k)$ and $y\not\in N(x'_1)\cup \dots \cup N(x'_k)$.
\end{lemma}
\begin{proof}
    Let $I = [\lceil \log(|S|) \rceil]^{k^2}$ and let $B = \{0,1\}^{k^2}$. Now define $J_S := I\times B$.
    Observe that $|J_S| = \bigO(2^{k^2}\log^{k^2}n) = \Tilde \bigO (1)$.
    
    Let $Y' = Y\cup J_S$ and add edge between $x\in X$ and $(I,B)\in J_S$ if and only if $(x,I,B)\in T$.
    Now fix any tuple of vertices $x_1,\dots, x_k\in S$. 
    By the previous lemma, for any $x'_1,\dots, x'_k\in S$, we can find a vertex $y\in J_S$ corresponding to some tuple of indices $I$ and bits $B$ such that $y\in N(x_i)$ and $y\not \in N(x'_i)$ for all $i$.
    Observe that since $y$ is adjacent only to the vertices in $S$, by choosing a vertex $x'_i\in X\setminus S$, this property remains preserved.
\end{proof}
Given a bipartite graph $G = (X,Y,E)$ and a subset $S$, the set of $\Tilde{\bigO}(1)$ vertices as described in the last lemma is called \emph{identifier gadget} and is denoted by $J_{S}$.
Constructing a graph $G' = (X, Y\cup J_{S}, E)$ from $G$ by adding the edges as described above will be referred to as \emph{attaching} the identifier gadget to $S$.
If the identifier $J_S$ is attached to $S$, we schematically represent this by adding a directed edge from $S$ to $J_S$. 

We proceed to give the reduction from any Multichromatic Basic Problem to the corresponding Bichromatic Basic Problem and then using the result from the beginning of the section, we conclude that all the Bichromatic Basic Problems are hard under SETH.
\begin{proposition}
    Let $\Phi$ be any basic property.
    Given an instance $G = (V_1,\dots, V_{k}, W)$ of the Multichromatic Basic Problem corresponding to $\Phi$ with $\bigO(n)$ vertices with $|W| = \Tilde{\bigO}(1)$, we can construct in time $\Tilde\bigO(n)$ an equivalent instance $G' = (X,Y,E)$ of the corresponding Bichromatic Basic Problem with $|Y| = \Tilde{\bigO}(|W|)$.
\end{proposition}
\begin{proof}
    Without loss of generality, assume that $\Phi= E(v_1,w)\lor \dots \lor E(v_\ell, w) \lor \overline{E(v_{\ell+1},w)}\lor\dots \lor \overline{E(v_k,w)}$ and let $\ell' := k-\ell$ denote the number of negative literals in $\Phi$.
    Consider the following construction of the graph $G' = (X, Y, E)$.
        For all $S'\subseteq \{V_1,\dots, V_k\}$ of size $k-1$, let $S = \bigcup_{s\in S'}s$ and attach an identifier gadget $J_S$ to $S$.
        If for some $1\leq i\leq \ell$, $V_i\not \in S'$, add all edges between $V_i$ and $J_S$.
        Let $X = V_1\cup \dots \cup V_{k}$ and let $Y$ consist of $W$ and all the gadgets $J_S$ as constructed above.
        This reduction is depicted schematically on Figure \ref{fig:multi-bi}.

    \begin{claim}
        If there are vertices $x_1,\dots, x_k\in X$ such that $N(x_{\ell+1})\cap \dots \cap N(x_k) \subseteq N(x_1) \cup \dots \cup N(x_\ell)$, then each of the sets $V_1,\dots, V_k$ contains exactly one of the vertices $x_1,\dots, x_k$.
    \end{claim}
    \begin{proof}
        Assume that for some $\ell+1\leq i\leq k$ the set $V_i$ contains none of the vertices $x_1,\dots, x_k$. 
        Observe that the identifier gadget $J_{X\setminus V_i}$ is attached to the whole set $X$ except $V_i$, hence contradicting Lemma \ref{lemma:identifier-gadget}.

        Assume now that for some $1\leq i\leq \ell$, $V_i$ contains none of the vertices $x_1,\dots, x_k$.
        By Lemma \ref{lemma:identifier-gadget}, there exists a vertex $v\in J_{X\setminus V_i}$ adjacent to all $x_{\ell+1},\dots, x_{k}$ and to no other vertex in $X\setminus V_i$.
        In particular, it is non-adjacent to $x_j$ for any $j\in [\ell]$, yielding a contradiction.
        Hence, each set $V_i$ contains at least one vertex and since we are placing $k$ vertices in $k$ sets, the upper bound is also trivially satisfied.
    \end{proof}
    \begin{claim}
         If there are vertices $x_1,\dots, x_k \in X$ such that $N(x_{\ell + 1})\cap \dots \cap N(x_k) \subseteq N(x_1) \cup \dots \cup N(x_\ell)$, then for every $1\leq i \leq \ell$, there is a $1\leq j\leq \ell$ such that $x_i\in V_j$.
    \end{claim}
    \begin{proof}
        Assume for contradiction that there exists some index $i$ such that $x_i\in X\setminus \bigcup_{1\leq j\leq \ell} V_j$.
        Then by the pigeonhole principle, there exists a set $V_j$ (for $1\leq j \leq \ell$) that does not contain any of the vertices $x_1,\dots, x_\ell$ and by the previous claim, $V_j$ contains $x_k$ (without loss of generality).

        Consider the identifier gadget $J_{X\setminus V_j}$. 
        By construction, every vertex in this gadget is adjacent to $x_k$.
        Hence, by the assumption that $N(x_{\ell + 1})\cap \dots \cap N(x_k) \subseteq N(x_1) \cup \dots \cup N(x_\ell)$, it must also hold that $N(x_{\ell+1})\cap \dots \cap N(x_{k-1})\cap J_{X\setminus V_j} \subseteq N(x_1) \cup \dots \cup N(x_\ell)$, thus contradicting Lemma \ref{lemma:identifier-gadget}.
    \end{proof}

    Now assume $V_1,\dots, V_k, W$ is a yes-instance of the multichromatic basic problem. 
    Namely, we can find $v_1\in V_1, \dots, v_k \in V_k$, such that for every $w\in W$, the formula $E(v_1,w) \lor \dots \lor E(v_{\ell},w)\lor \overline{E(v_{\ell+1},w)}\lor \dots \lor \overline{E(v_{k},w)}$ is satisfied. 
    We claim that in $G'$,
    for every $y\in Y$, the formula $E(v_1,y) \lor \dots \lor E(v_{\ell},y)\lor \overline{E(v_{\ell + 1},y)}\lor \dots \lor \overline{E(v_{k},y)}$ is satisfied. 
    Let $v\in Y$ be arbitrary.
    If $v\in W$, then it is trivially satisfied.
    Consider now $v\in Y\setminus W$. 
    Clearly $v\in J_{X\setminus V_i}$ for some $i$.
    If $1\leq i\leq \ell$, then by construction, all edges between $V_i$ and $J_{X\setminus V_i}$ are present and in particular $v$ is adjacent to the positive vertex $v_i$.
    On the other hand if $\ell+1\leq i \leq k$, then there are no edges between $V_i$ and $J_{X\setminus V_i}$ and in particular $v$ is nonadjacent to the negative vertex $v_i$.
    In both cases, $v$ is satisfied.
    
    Conversely, assume that there are vertices $x_1,\dots, x_k\in X$ such that for every $y$ in $Y$ the formula $E(x_1,y) \lor \dots \lor E(x_{\ell},y)\lor \overline{E(x_{\ell+1},y)}\lor \dots \lor \overline{E(x_k,y)}$ is satisfied.
    By the two claims above (without loss of generality), $x_1\in V_1, \dots, x_k\in V_k$.
    And now it is sufficient to observe that since $W\subseteq Y$ and every vertex in $Y$ is satisfied by $x_1,\dots, x_k$, every vertex in $W$ is also trivially satisfied by $x_1,\dots, x_k$.
\end{proof}
\begin{figure}
    \centering
    \includegraphics[scale=0.6]{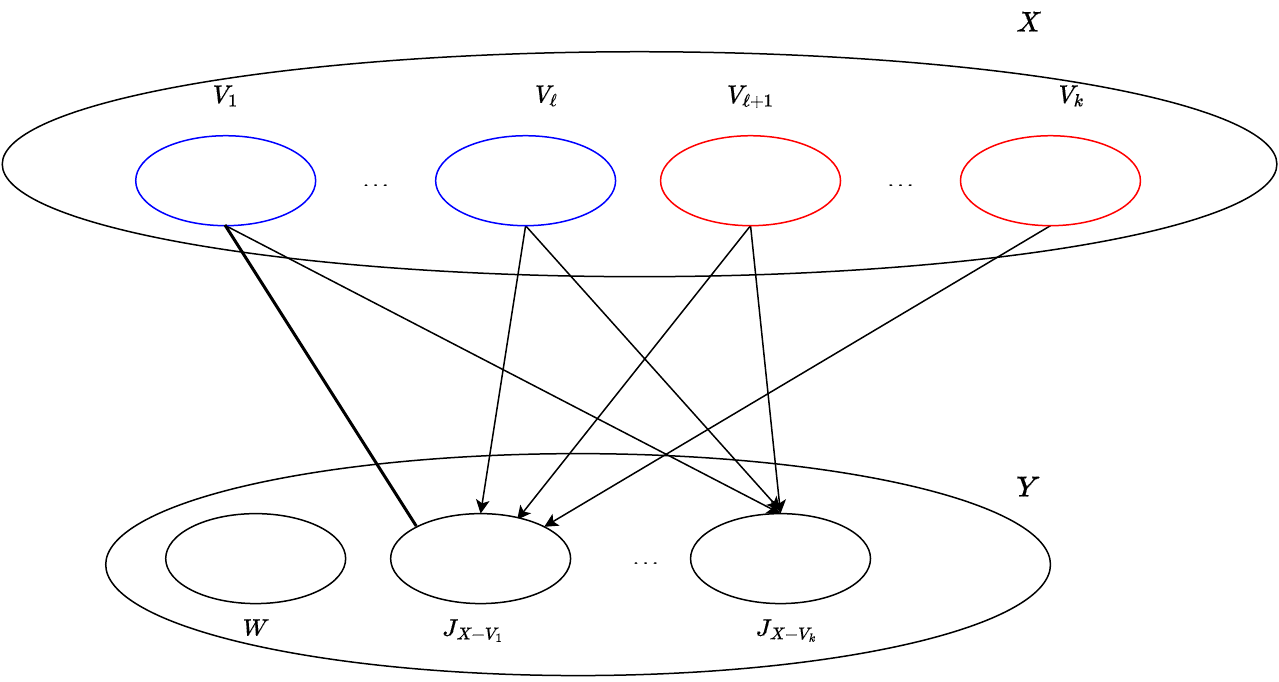}
    \caption{Schematical representation of the reduction from Multichromatic Basic Problem to Bichromatic Basic Problem. Blue sets contain positive vertices, red sets contain negative vertices, bold edges represent bicliques between the corresponding sets and directed edge from set $V_i$ to the identifier gadget $J_S$ represents attaching the gadget to the set. The edges from $V_i$ to $W$ are inherited from the multichromatic instance.}
    \label{fig:multi-bi}
\end{figure}
The following proposition now follows directly.
\begin{proposition}
    Let $\Phi$ be a basic property with at least one negative literal. 
    For any $\varepsilon>0$, the Bichromatic Basic Problem corresponding to $\Phi$ cannot be solved faster than $\bigO(n^{k-\varepsilon})$ even on graphs with $\Tilde \bigO(n)$ edges, unless SETH fails.
\end{proposition}

\subsection{Bichromatic to Monochromatic}
We now proceed to prove the hardness of Monochromatic Basic Problems in sparse graphs with at least one negative literal in the corresponding basic property by reducing from the corresponding Bichromatic Problem.

\begin{lemma}\label{lemma:monochromatic-gadget}
    Given any positive integer $N$ and a constant $k$, there exists a graph $G$ with at most $\Tilde{\bigO}(N)$ vertices satisfying the following conditions. 
    \begin{itemize}
        \item Every vertex $v$ is labelled by an integer $\lambda(v)\in \{1,\dots, N\}$.
        \item For all $a_1,\dots, a_k, b_1,\dots, b_k\in V(G)$ and for every $i\in 1,\dots, N$, there exists a vertex $x\in V(G)\setminus \{a_1,\dots, a_k,b_1,\dots, b_k\}$ such that $\lambda(x) = i$ and $x$ is adjacent to all $a_j$ and nonadjacent to all $b_j$ for $j\in [k]$. 
        \item We can compute $G$ in deterministic polynomial time in $N$.
    \end{itemize} 
\end{lemma}
Before proving this lemma, we first show how one can use this result to obtain a reduction from a sparse Bichromatic Basic Problem to the corresponding Monochromatic Basic Problem.
\begin{proposition}
    Let $\Phi$ be a basic property with $k$ variables and at least one negative literal.
    Given an instance $G = (X, Y, E)$ of the Bichromatic Basic Problem corresponding to $\Phi$ with $n$ vertices and assume $Y=\Tilde{\bigO}(1)$, we can construct an equivalent instance $G' = (V,E)$ of the corresponding Monochromatic Basic Problem in time $\Tilde\bigO(n)$ consisting of at most $\Tilde{\bigO}(n)$ edges.
\end{proposition}
\begin{proof}

    Without loss of generality, assume that ${\Phi= E(x_1,y)\lor \dots \lor E(x_\ell, y) \lor \overline{E(x_{\ell+1},y)}\lor \dots \lor \overline{E(x_k,y)}}$, for some $\ell<k$, and let $k-\ell = \ell'\geq 1$.
    Consider the graph $G'$ constructed as follows.

        Let $N = |Y|$ and let $H$ be a graph satisfying the properties of Lemma \ref{lemma:monochromatic-gadget}. Without loss of generality, assume that $Y = \{y_1,\dots, y_N\}$ and label the vertex $y_i$ by $i$.        
        Let $V(G') = V(H) \cup X$. Let $H \cong G[V(H)]$, and add an edge between $x\in X$ and $h\in H$ if and only if $\lambda(h) = i$ and there is an edge between $x$ and a vertex in $Y$ labeled $i$ in $G$. 
    
    First assume $G$ is a yes-instance of the Bichromatic Basic Problem. 
    Namely, assume that we can find $x_1, \dots, x_k\in X$ such that for every $y\in Y$, the formula $E(x_1,y) \lor \dots \lor E(x_{\ell},y)\lor \overline{E(x_{\ell + 1},y)}\lor \dots \lor \overline{E(x_{k},y)}$ is satisfied.
    We claim that also $\forall v\in V(G')$, it holds that $ E(x_1,v) \lor \dots \lor E(x_{\ell},v)\lor \overline{E(x_{\ell + 1},v)}\lor \dots \lor \overline{E(x_{k},v)}$. 
    To see this assume first that $v\in X$. 
    Since there are no edges inside $X$, clearly each vertex is nonadjacent to e.g. $x_{\ell + 1}$ hence satisfying the desired formula. 
    
    Now assume that $v\in V(H)$. 
    Then $\lambda(v) = i$ for some $1\leq i\leq N$. 
    Consider the vertex $y$ in $Y\subseteq V(G)$ that is labelled $i$. 
    Since $G$ is a yes-instance, this vertex is either adjacent to $x_j$ for some $j\in [\ell]$, or nonadjacent to $x_j$ for some $\ell+1\leq j\leq k$.
    Now, since $v$ and $y$ have the same label in the respective graph, by construction of $G'$, if $E(x_j,y)$, then also $E(x_j, v)$ and similarly if $\overline{E(x_j,y)}$, then $\overline{E(x_j,v)}$. 
    Therefore the desired formula is satisfied for $v$ as well.
    We have thus covered all the vertices from $V(G')$.

    On the other hand, assume $G$ is a no-instance of the Bichromatic Basic Problem. 
    In particular, assume that for every selection of $x_1,\dots, x_k$, there exists $y\in Y$ that is non-adjacent to all $x_1,\dots, x_\ell$ and adjacent to all $x_{\ell+1},\dots, x_k$.
    Consider any collection of vertices $v_1,\dots, v_k\in V(G')$. 
    Let $S_X$ be the subset of $v_1,\dots, v_k$ consisting of all vertices $v_j$ contained in $X$.
    Clearly, we can find a vertex $y_i\in Y$ that is not satisfied by $S_X$.
    By construction of $G'$ it follows that any vertex $w\in V(H)$ with $\lambda(w) = i$ is not satisfied by $S_X$.
    
    Moreover, by Lemma \ref{lemma:monochromatic-gadget}, there exists a vertex $w\in V(H)$ with $\lambda(w) = i$ such that $w$ is not satisfied by any vertex $v_i$ in $\{v_1,\dots,v_k\}\setminus S_X$.
    In particular, it follows that such $w$ is not satisfied by $v_1,\dots, v_k$ and since the vertices $v_1,\dots, v_k$ were selected arbitrarily, $G'$ is also a no-instance.
\end{proof}

We now proceed to prove the Lemma \ref{lemma:monochromatic-gadget}. 
But before that, we state a few basic well known facts about polynomials that will be useful.
\begin{lemma} \label{lemma:polynomials-f-g}
    For any two polynomials $f,g$ with coefficients in a field $\F$ of degree at most $d$, there are at most $d$ points $x$ in $\F$ such that $f(x) = g(x)$.
\end{lemma}
\begin{lemma} \label{lemma:interpolation-uniqueness}
    Given $d+1$ distinct points $x_1,\dots, x_{d+1}\in \F$ and $d+1$ values $y_1,\dots, y_{d+1}\in \F$, there is a unique polynomial $f\in \F[X]_{\leq d}$ satisfying $f(x_i) = y_i$ for all $i$.
\end{lemma}

\begin{proof}[Proof (of Lemma \ref{lemma:monochromatic-gadget}).]
    Let $p$ be the smallest prime number strictly larger than $k^2(\log(N) + k + 2)$ and let $\F_p$ denote the finite field of order $p$.
    Let $d = \lceil \log_p(N)\rceil + k$. 
    Let ${V(G) = \F_p[X]_{\leq d}}$, that is each vertex in $G$ corresponds to a polynomial with coefficients in $\F_p$ of degree at most $d$.
    Let $L = \lceil \log_p(N)\rceil$ and for $f\in V(G)$ let $\lambda(f) = i$ if and only if $\psi(f(1), \dots, f(L)) = i$ be the labeling of the vertices in $G$ for a bijective function $\psi:\F_p^L \to [p^L]$ (if $\psi(f(1),\dots, f(L))>N$, then we relabel $f$ to $N$).
    We add the edges to our graph as follows.
    Let $R = L + dk + k + 1$ and for any $f,g\in V(G)$, let $\{f,g\}\in E(G)$ if and only if there exists an $x$ such that $L < x\leq R$ and $f(x) = g(x)$. 
    We observe that $p>R$, hence the elements $1,\dots, R$ are all distinct in $\F_p$. 

    We proceed to show that given any distinct polynomials $a_1,\dots, a_k$, $b_1,\dots, b_k$ and any label $i$, we can find a polynomial $f$ with $\lambda(f) = i$ such that for all $j\in [k]$ there exists $L+1\leq x \leq R$ satisfying $a_j(x) = f(x)$, and for no $j\in [k]$ is there an $L+1\leq x\leq R$ satisfying $b_j(x) = f(x)$.
    In particular, note that the vertex corresponding to such polynomial would be adjacent to all the vertices corresponding to $a_j$ and non-adjacent to all vertices corresponding to $b_j$ and would thus yield the desired property.
    
    \begin{claim}
        For any fixed set of polynomials $\{a,b_1,\dots, b_k\}\subseteq V(G)$ and a $k-1$ element set $A\subseteq \{L+1,\dots, R\}$ there exists a point $x\in \{L+1,\dots, R\} \setminus A$ such that $a(x)\not \in \{b_1(x),\dots, b_k(x)\}$
    \end{claim}
    \begin{proof}
        By Lemma \ref{lemma:polynomials-f-g}, $a$ agrees with $b_i$ in at most $d$ points.
        Hence, in total, there are at most $dk$ points $x$ such that $a(x)\in \{b_1(x), \dots, b_k(x)\}$.
        
        Consequently, there are at least $R-(L+1)-dk = k$ points $x$ such that $a(x)\not\in \{b_1(x),\dots, b_k(x)\}$ thus proving the desired.
    \end{proof}
    
    Let $a_1,\dots, a_k, b_1,\dots, b_k$ be arbitrary vertices.
    By the claim above, we can find distinct elements $x_1,\dots, x_k$ such that $a_i(x_i) \not \in \{b_1(x_i),\dots, b_k(x_i)\}$ for any $1\leq i \leq k$.
    Let $\ell$ be an arbitrary fixed label such that $\psi^{-1}(\ell) = (\ell_1,\dots, \ell_L)$.
    Consider the set $S$ consisting of polynomials $f$ satisfying $f(1) = \ell_1, \dots, f(L) = \ell_L$ and $f(x_i) = a_i(x_i)$ for every $1\leq i \leq k$.
    Since $d = L+k$, by Lemma \ref{lemma:interpolation-uniqueness}, if we fix a point $t\in \{L+1,\dots, R\}\setminus\{x_1,\dots, x_k\}$, there is a unique polynomial $f$ in $S$, such that $f(t) = b_j(t)$ for any $1\leq j\leq k$.
    Hence, the set 
    $$T = \{f\in S \mid \exists L+1\leq t\leq R \; \exists 1\leq j\leq k \text{ such that } f(t) = b_j(t) \}$$ 
    contains at most $(R-L-k-1)k$ elements.

    Similarly, by Lemma \ref{lemma:interpolation-uniqueness}, there are at least $p$ distinct polynomials in $S$.
    Notice that it is sufficient to show that the set $S\setminus T$ is non-empty, as this would prove the existence of a vertex adjacent to all $a_i$ and non-adjacent to all $b_j$ for all $i,j$.
    To this end, we compute the value of
    $$|S\setminus T|\geq p - (R-L-k-1)k > p - R > 0.$$

    Construction of such a graph can be implemented in polynomial time by evaluating each of the $\bigO(p^d)$ polynomials in at most $\bigO(p)$ points and then iterating through $\bigO(p^{2d})$ pairs of vertices and checking in time $\bigO(p)$ if they match in the corresponding coordinates.
    The total time complexity of the construction is thus at most
    $$\bigO(dp^{d+1} + dp^{2d})\leq \bigO(p^{3\log_p(N)}) = \bigO(N^3).$$
    
    Finally, we have to argue that $|V(G)| \leq \Tilde{\bigO}(N)$.
    This, however follows from the simple observation that there are at most $p^{d+1}$ distinct polynomials in $\mathbb{F}_p[X]_{\leq d}$ and 
    $$p^{d+1} \leq  p^{\log_p(N) + k + 2} = Np^{k+2}\leq \Tilde{\bigO}(N).$$
\end{proof}

Combining the results from the last two subsections, we obtain the main result of this chapter.
\begin{theorem}
    Let $\Phi$ be a basic property with at least one negative literal. 
    For any $\varepsilon>0$, the Monochromatic Basic Problem corresponding to $\Phi$ cannot be solved faster than $\bigO(n^{k-\varepsilon})$ even on graphs with $\Tilde \bigO(n)$ edges, unless SETH fails.
\end{theorem}

\section{Algorithms and Hardness of Distance-\texorpdfstring{\boldmath{$r$}}{r} \texorpdfstring{\boldmath{$k$}}{k}-Dominating Set}
\label{sec:distance-r-domination}

Let $r \geq 1$ be a fixed integer.
We consider the \emph{Distance-$r$ $k$-Dominating Set} problem as a natural generalization of $k$-Dominating Set problem, where a vertex $u$ dominates a vertex $v$ if and only if there exists a path from $u$ to $v$ of length at most $r$.
Clearly, if $r = 1$, we obtain the usual $k$-Dominating Set problem.

Formally, Distance-$r$ $k$-Dominating Set problem can be stated as follows.
Given an undirected graph $G=(V,E)$ with $|V|=n$ and $|E|=m$, decide if
\[ \exists v_1\in V \dots \exists v_k \in V: \forall w \in V: \min \{d_G(v_1,w), \cdots, d_G(v_k,w)\} \leq r.\} \]
Let $B_r(v) = \{x\in V(G)\mid d(x,v)\leq r\}$.

In this section, we prove that sparsity does not affect the fine-grained complexity of Distance-$r$ $k$-Dominating Set problem.
In fact, we prove that a lower bound $n^{k\pm o(1)}$ holds already for graphs with $m = \Tilde{\bigO}(n)$ edges, assuming the $k$-OV hypothesis.
Furthermore, we prove that if $\omega = 2$, or $k$ is sufficiently large, we can construct an algorithm that runs in $\tilde{\bigO}(n^{k+o(1)})$, thus matching the lower bound.

\label{sec:domination-distance}
\begin{theorem}
    For any $\varepsilon>0$, there exists no algorithm solving Distance-$r$ $k$-Dominating Set problem in $\bigO(n^{k-\varepsilon})$ for any $k\geq 2, r\geq 2$, even in sparse graphs, unless SETH fails.
\end{theorem}
\begin{proof}
    To prove this lower bound, we reduce from $k$-OV.
    Let $A_1,\dots, A_k$ be an instance of $k$ orthogonal vectors problem.
    In particular, let $A_1,\dots, A_{k}$ denote sets of $d$-dimensional (assume $d=\bigO(\log n)$) binary vectors of size $n$.
    Construct the graph $G$ as follows.
        Start with an empty graph.
        For each set $A_i$, add one vertex for every vector in $A_i$, denoted $a^i_j$ and two additional vertices $A_i$, $A_i'$.
        Add an edge between $A_i$ and $a^i_j$ and between $A'_i$ and $a^i_j$ for every $i,j$.
        Add an edge between $A_i$ and $A_i'$ and subdivide this edge $r-2$ times and label the new vertices by $s_1(A_i,A_i'), \dots, s_{r-2}(A_i,A_i')$.
        Add $2d$ vertices $D_1,\dots, D_d$, $D_1',\dots, D_d'$ and add an edge between $a^i_j$ and $D_t$ if and only if $a^i_j[t] = 0$ for all $i,j,t$. 
         Add an edge between $D_t$ and $D_t'$ for all $t$, and subdivide this edge $r-2$ times. 
        Label the new vertices by $s_1(D_t, D_t'), \dots, s_{r-2}(D_t,D_t')$.
    The construction of such a graph is schematically depicted on Figure \ref{fig:dom-distance}.

    Clearly the constructed graph $G$ has $\Tilde\bigO(n)$ vertices and $\Tilde\bigO(n)$ edges.
    We proceed to show that $G$ contains $k$ vertices $v_1,\dots, v_k$ such that $\bigcup_{i=1}^k B_r(v_i) = V(G)$ if and only if there are vectors $a_1\in A_1,\dots, a_k\in A_k$ such that $\prod_{i=1}^k a_i[j] = 0$ holds for all $1\leq j\leq d$.

    Suppose first that there are vectors $a_1\in A_1,\dots, a_k\in A_k$ such that $\prod_{i=1}^k a_i[j] = 0$ for all $1\leq j\leq d$.
    Then we claim that every vertex in $G$ is at distance at most $r$ from at least one $a_i$. 
    Consider first the vertices $A'_i$ and $A_i$.
    By construction $A_i$ is adjacent to $a_i$ and $A_i'$ is at distance exactly $r$ from $a_i$. 
    Consequently also $s_1(A_i, a_i), \dots, s_{r-1}(A_i, a_i)$ are covered.
    Furthermore, since any vertex $a^i_j$ is adjacent to $A_i$, $d(a_i,a^i_j)\leq 2$ for all $i,j$ and since $r\geq 2$, also vertices $a^i_j$ are covered.  
    Finally, by construction, for any $1\leq t\leq d$ there is an edge from some $a_i$ to $D_t$, therefore since $d(D_t,D'_t) = r-1$, we conclude that also all the vertices along the (unique) path from $D_t$ to $D'_t$ are covered and thus $a_1,\dots, a_k$ cover all the vertices in $G$.

    Conversely assume that there are vertices $v_1,\dots, v_k$ such that $\bigcup_{i=1}^k B_r(v_i) = V(G)$.
    Consider vertices $A'_1,\dots, A_k'$. 
    Clearly $B_r(A_i')$ consists only of those vertices along the paths from $A_i'$ to $a^i_j$ for all $j$.
    In particular, since for all $i\neq j$ $B_r(A_i')$ and $B_r(A_j')$ are pairwise disjoint, the only possible solution is $v_i\in B_r(A_i')$ (w.l.o.g.). 
    Furthermore, for any $1\leq t\leq d$, $B_r(D_t')$ consists only of the vertices along the path from $D_t'$ to some $a^i_j$, such that $a^i_j$ is adjacent to $D_i$.
    Considering also the last constraint, we can conclude that for all $t$, there is a vertex in $\{v_1,\dots, v_k\}$ that is adjacent to $D_t$.
    
    Hence, for every $v_i = a^i_j$ (for some $j$), take $a_i$ to be the vector that corresponds $a^i_j$, and for $v_i$ that lies along the path from $A_i$ to $A_i'$, take $a_i$ that corresponds to $a^i_1$, then we claim that $a_1\in A_1,\dots, a_k\in A_k$ and for every $1\leq t\leq d$ $\prod_{i=1}^k a_i[d] = 0$. 
    By selection of the vectors $a_i$, clearly $a_i\in A_i$ and also, in $G$, if $v_i$ is adjacent to $D_t$, then $v_i = a^i_j$ for some $j$ and furthermore, $a_i[t]=0$.
    But as we argued above, for every $t$ such a vertex $v_i$ exists and so $a_i[t]=0$.
    The desired result follows directly.
\end{proof}
\begin{figure}
    \centering
    \includegraphics[scale=0.9]{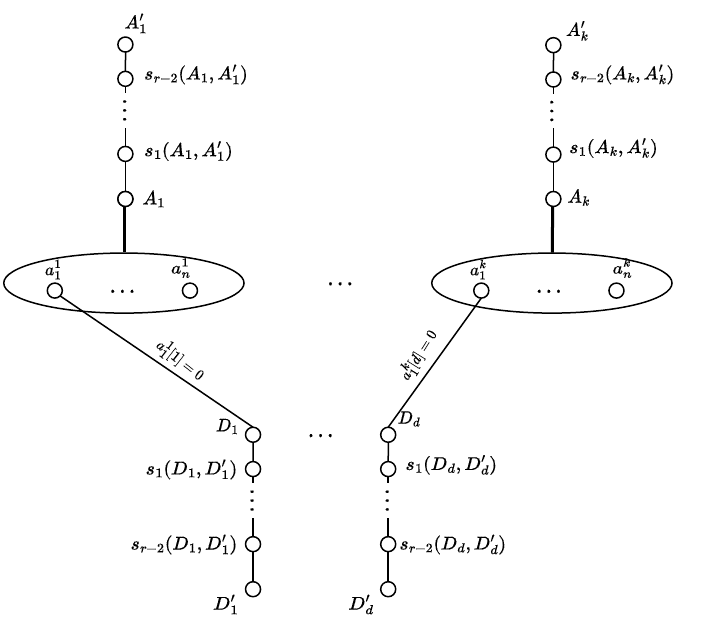}
    \caption{Schematic representation of reduction from $k$-OV to Distance-$r$ $k$-Dominating Set problem. Bold edges represent bicliques between the corresponding sets.}
    \label{fig:dom-distance}
\end{figure}

\begin{proposition}
    Let $G$ be a graph with $n$ vertices and $m$ edges and let $k\geq 2,r\geq 1$ be fixed constants.
    For sufficiently large $k$, or $\omega=2$, we can solve Distance-$r$ $k$-Dominating Set in time $\bigO(n^{k+o(1)})$.
\end{proposition}
\begin{proof}
    Without loss of generality, we may assume that $r>1$, as otherwise we have an instance of $k$-dominating set problem and we already know that this statement holds.
    Let $A$ be the adjacency matrix of $G$. 
    Compute the matrix $(A+I)^r$.
    We can compute this matrix in $\bigO(n^{\omega+o(1)} \log(r)) = \bigO(n^{\omega+o(1)})$ time (by repeated squaring of matrix $A+I$). 
    The entry $i,j$ in the matrix $(A+I)^r$ is non-zero if and only if there is a walk of length at most $r$ between $i$ and $j$ in $G$.
    Observe that every path is a walk, so if there exists a path of length at most $r$ between $i$ and $j$, the entry $(A+I)^r_[i,j]$ will be non-zero.
    Furthermore, if there exists a walk of length at most $r$ between $i$ and $j$, then there exists a path of length at most $r$ between them as well (obtained by removing all the cycles from the walk).
    
    Let $B$ be a matrix obtained by normalizing $(A+I)^r$.
    In particular, $B[i,j] = 1$ if and only if $(A+I)^r[i,j]>0$, otherwise, $B[i,j] = 0$.
    and let $G'$ be a graph with adjacency matrix $B - I_n$.
    Note that $V(G) = V(G')$.
    We claim that $G'$ has dominating set at most $k$ if and only if $G$ has a dominating set of size $k$ at distance $r$.
    
    In particular, assume first that $G$ has a Distance-$r$ $k$-Dominating Set, $S = \{v_1,\dots, v_k\}$.
    Fix a vertex $v\in V(G)$. 
    We want to show that $v$ is dominated by some $v_i$ in $G'$.
    If $v\in S$, we are done, so assume this is not the case.
    Then for some $1\leq i\leq k$, there exists a (non-trivial) path of length $t\leq r$ between $v_i$ and $v$ in $G$.
    Hence, the entry corresponding to $v_i,v$ in $(A+I)^r$ will be non-zero and since $v_i\neq v$, this entry will be equal to $1$ in $B-I_n$ and therefore $v_i$ is adjacent to $v$ in $G'$.
    Since $v$ was an arbitrary vertex, we may conclude that $S$ dominates $G'$.

    Conversely, assume that $S= \{v_1,\dots, v_k\}$ is a $k$-dominating set in $G'$ and fix a vertex $v\in V(G)$.
    Then $v$ is either contained in $S$, or adjacent to a vertex in $S$ in $G$.
    If $v$ was contained in $S$, then there exists a path of length $0<r$ between $v$ and a vertex in $S$ and $v$ is dominated by $S$.
    Otherwise, assume that $v$ is adjacent to a vertex $v_i\in S$. 
    Then there exists a walk of length $t\leq r$ in $G$ and in particular, a path of length at most $t$. 
    Hence, $v$ is dominated by $v_i$ and $S$ dominates $G$ as desired.
    
    Observe that we can compute $B$ in $\bigO(n^{\omega+o(1)})$, and as we proved in section \ref{sec:k-dom-set}, we can check if a graph $G'$ has a dominating set of size $k$ in $\bigO(n^{k+o(1)})$, assuming $k\geq 8$, or $k\geq 2$ and $\omega = 2$.
\end{proof}

\section{Closed Neighbourhood Containment and \texorpdfstring{\boldmath{$(k-1)$}}{(k-1)}-Covering}
In this section we show that if we assume that the edge relation is reflexive (i.e., $E(v,v)$ for all $v\in V$), then the basic problem 
\[\exists \text{pairwise distinct } v_1,\dots, v_k: \forall w \in V: \overline{E(v_1,w)} \vee E(v_2,w)\vee \cdots \vee E(v_k,w)\]
becomes easier on sparse graphs. For $k=2$, we call this problem \emph{Closed Neighborhood Containment} (are there distinct $v_1,v_2\in V$ such that $N[v_1] \subseteq N[v_2]$?). 
For general $k\geq 2$, we call it \emph{Closed Neighborhood $(k-1)$-Covering} (are there pairwise distinct $v_1,\dots, v_k\in V$ such that $N[v_1] \subseteq N[v_2]\cup \cdots \cup N[v_k]$?).
 
\paragraph{Closed Neighborhood Containment.}
To obtain faster algorithms on sparse graphs for $k=2$, we exploit the following simple, yet crucial observation:
\begin{observation}
    Given a graph $G$, if there exists a pair of vertices $x,y$, such that $N[x] \subseteq N[y]$, then $E(x,y)$.
\end{observation}
We can use this fact to reduce the problem to the problem All-Edge Triangle Counting (Given an $m$-edge graph, for each edge $e\in E$, output the number of triangles containing $e$.), which can be solved in time $\bigO(m^{\frac{2\omega}{w+1}})$~\cite{AlonYZ97}.
\begin{lemma}
    Given a graph $G$, there exists a pair of vertices $x,y$, such that $N[x] \subseteq N[y]$ if and only if $E(x,y)$ and the number of triangles containing the edge $\{x,y\}$ is equal to $\deg(x)-1$.
\end{lemma}
\begin{proof}
    Suppose first that $N[x]\subseteq N[y]$.
    Then by the previous lemma $E(x,y)$. 
    Clearly, the number of triangles containing the edge $\{x,y\}$ is at most $\deg(x)-1$, so it is sufficient to show that the lower bound holds as well.
    Assume for contradiction that this number is strictly smaller than $\deg(x)-1$.
    Then there exists $z\in N(x)$ such that $z\not \in N(y)$, contradicting the assumption that $N[x]\subseteq N[y]$.

    Conversely, assume that for some pair of adjacent vertices $x,y$ the number of triangles containing the edge $\{x,y\}$ is equal to $\deg(x)-1$.
    Then clearly every vertex $z\neq y$ that is adjacent to $x$ is contained in this triangle.
    In particular every $z\in N(x)\setminus\{y\}$ is adjacent to $y$ as desired.
\end{proof}

\begin{theorem}
    Closed Neighborhood Containment can be solved in $\bigO(m^{\frac{2\omega}{\omega+1}})$.
\end{theorem}
\begin{proof}	
    By computing, for each $e\in E$, the number of triangles containing $e$ in total time $\bigO(m^{\frac{2\omega}{\omega+1}})$, we can use the last lemma to decide any given instance in the same running time.
\end{proof}

\paragraph{Closed Neighborhood \boldmath{$(k-1)$}-Covering} The above algorithm easily generalizes to $k>2$ by brute-forcing over $k-2$ variables as follows:

Given a graph $G$ with $n$ vertices and $m$ edges, fix distinct $x_1,\dots, x_{k-2}$ arbitrarily.
Construct a tripartite graph $G'$ as follows.
Let $V_1 = V_2 = V(G)\setminus \{x_1,\dots, x_{k-2}\}$ and let $V_3 = V(G)\setminus \big(N[x_1]\cup \dots \cup N[x_{k-2}] \big)$.
Let $V(G') = V_1\cup V_2\cup V_3$.
Finally, add edges in $G'$ as follows. 
If a pair $u\neq v$ is in $E(G)$, add the edge between the vertex corresponding to $u\in V_1$ and the vertex corresponding to $v\in V_2$. 
Add the edges between $V_3$ and $V_1$, and $V_3$ and $V_2$ similarly.
Finally, for any vertex $v$ in $V_3$, add the edge between the copy of $v$ in $V_3$ and the copy of $v$ in $V_1$ and $V_2$. 
For each vertex $x\in V_1\cup V_2$ store the number $\pi(x)$ of neighbours of $x$ in $V_3$. 
Observe first that $G'$ has $\bigO(n)$ vertices and $\bigO(m)$ edges.

\begin{lemma}
	Given a graph $G$, and a set of distinct vertices $x_{1},\dots, x_{k-2}$, there exists a pair of distinct vertices $x_{k-1}, x_k$, such that $N[x_k]\subseteq \bigcup_{i=1}^{k-1}N[x_i]$ and $E(x_{k-1},x_k)$ if and only if in the graph $G'$ as constructed above the number of triangles containing $x_{k-1}\in V_1$ and $x_{k}\in V_2$ is equal to $\pi(x_k)$.
\end{lemma}
\begin{proof}
	Suppose first that there exists a pair of distinct vertices $x_{k-1}, x_k$, such that $N[x_k]\subseteq \bigcup_{i=1}^{k-1}N[x_i]$ and $E(x_{k-1},x_k)$.
	Consider the copy of $x_{k-1}$ in $V_1$ and the copy of $x_k$ in $V_2$. 
	
	Fix any vertex $v$ in $V_3$ and consider two cases.
	\begin{itemize}
		\item $v$ corresponds to the same vertex as $x_{k-1}$, or $x_k$ in $G$.
		\item $v$, $x_{k-1}$, $x_k$ are all distinct vertices in $G$.
	\end{itemize}
	Notice that these are the only possibilities, since the endpoints of the edges going between $V_1$ and $V_2$ by construction correspond to distinct vertices in $G$.
	
	If $v \in \{x_{k-1}, x_k\}$, since $x_{k-1}$ and $x_k$ are distinct vertices, adjacent in $G$, in $G'$ the vertices $\{x_{k-1}\in V_1, x_k\in V_2, v\in V_3\}$ form a triangle.
	
	Assume now that $v$, $x_{k-1}$, $x_k$ are all distinct vertices in $G$.
	Since in $G$, $v$ is not in $\bigcup_{i=1}^{k-2}N[x_i]$, it is either in $N(x_{k-1})$, or is not in $N(x_k)$. 
	But by construction of $G'$ this property is preserved and hence if $v\in N(x_{k})$, then $v\in N(x_{k-1})$, or in other words for any vertex $v\in V_3$ that is adjacent to $x_k$, it holds that $v, x_{k-1}, x_k$ form a triangle, giving us the lower bound.
	The upper bound is trivial from the fact that $G'$ is tripartite.
	
	Conversely, assume that for a pair of vertices $x_{k-1}\in V_1$ and $x_k\in V_2$, the number of triangles in $G'$ containing these two vertices is equal to $\pi(x_k)$.
	Let $v$ be an arbitrary vertex in $G$.
	If $v$ is not in $V_3$ in $G'$, then it is contained in $\bigcup_{i=1}^{k-2} N[x_i]$ in $G$.
	On the other hand, if it is contained in $V_3$, in $G'$ it is either adjacent to $x_{k-1}\in V_1$, or nonadjacent to $x_k\in V_2$.
	In $G$, this means that $v$ is either equal to $x_k$ or $x_{k-1}$, or adjacent to $x_{k-1}$, or nonadjacent to $x_k$.
	In other words, if $v\in V_3$ in $G'$, then in $G$ it is either contained in $N[x_{k-1}]$, or not contained in $N[x_{k}]$, as desired.
\end{proof}

We can observe that the last lemma gives us an algorithm running in time $\bigO(n^{k-2}m^{\frac{2\omega}{\omega+1}})$, beating the $\Omega(n^k)$ running time for sufficiently sparse graphs.
\begin{theorem}
    Closed Neighborhood $(k-1)$-Covering can be solved in time $\bigO(n^{k-2}m^{\frac{2\omega}{\omega+1}})$.
\end{theorem}

\paragraph{Hardness of All Remaining Reflexive Properties}
We now argue that for every monochromatic basic problem 
\[\exists x_1\in V,\dots, \exists x_{\ell} \in V, \exists x_{\ell+1} \in V, \dots, \exists x_{k}\in V \forall v\in V \; \Big(\bigvee_{i=1}^{\ell}E(x_i,v)\Big)\lor\Big(\bigvee_{i=\ell+1}^{k}\overline{E(x_i,v)}\Big),\]
the reduction given in Section~\ref{sec:monochromatic-props} can be adapted for the reflexive case if the number of negative literals is $\ell'\geq 2$). 
Note that for $\ell' = 0$ ($k$-Dominating Set) and $\ell'=1$ (Closed Neighborhood $(k-1)$-Covering), we have shown faster algorithms in sparse graphs, so $\ell'\geq 2$ is indeed the only case remaining.
 
For $\ell'\geq 2$ we construct the reduction from a sparse Bichromatic Basic Problem to the corresponding Monochromatic Basic Problem similarly as in Section~\ref{sec:monochromatic-props}.
\begin{lemma}
    Let $\Phi$ be a basic property with $k$ variables and at least two negative literals.
    Given an instance $G = (X, Y, E)$ of the Bichromatic Basic Problem corresponding to $\Phi$ with $n$ vertices and assume $Y=\Tilde{\bigO}(1)$, we can construct an equivalent instance $G' = (V,E)$ of the corresponding Monochromatic Basic Problem in reflexive setting in time $\Tilde\bigO(n)$ consisting of at most $\Tilde{\bigO}(n)$ edges.
\end{lemma}
\begin{proof}
    Consider the same construction of the graph $G'$ as described in Section~\ref{sec:monochromatic-props}.
    If $G$ is a no-instance, one can observe that the identical arguments show that $G'$ is also a no-instance.
    
    On the other hand, if $G$ is a yes-instance, then we can find $x_1,\dots, x_k\in X$ such that for every $y\in Y$, the formula $E(x_1,y) \lor \dots \lor E(x_{\ell},y)\lor \overline{E(x_{\ell + 1},y)}\lor \dots \lor \overline{E(x_{k},y)}$ is satisfied.
    We claim that also $\forall v\in V(G')$, it holds that $ E(x_1,v) \lor \dots \lor E(x_{\ell},v)\lor \overline{E(x_{\ell + 1},v)}\lor \dots \lor \overline{E(x_{k},v)}$. 
    If $v\in V(H)$, then again the identical arguments as before show that $v$ is satisfied.
    Thus, we may assume that $v\in X$.
    For every $v\in X\setminus\{x_{\ell+1}\}$, we observe that $v$ is non-adjacent to $x_{\ell+1}$ and hence satisfied.
    However, since we are in the reflexive case, $x_{\ell+1}$ is adjacent to itself, and since there are no (non-trivial) edges inside $X$, $x_{\ell+1}$ has to be satisfied by another negative vertex (this is a subtle, but crucial reason why the reduction fails in the reflexive setting for $\ell'\leq 1$).
    Since $\ell'>1$, notice that $x_{\ell+1}$ is satisfied by $x_{\ell+2}$. 
\end{proof}
We conclude this section by stating the consequential hardness result.
\begin{theorem}
    Let $\Phi$ be a basic property with at least two negative literals. 
    For any $\varepsilon>0$, the Monochromatic Basic Problem in the reflexive setting corresponding to $\Phi$ cannot be solved faster than $\bigO(n^{k-\varepsilon})$, even on graphs with $\Tilde \bigO(n)$ edges, unless SETH fails.
\end{theorem}

\bibliographystyle{plainurl}
\bibliography{refs}

\end{document}